\documentclass{article}     

\usepackage{amsmath, amsfonts,amsthm,amssymb}
\usepackage{dsfont}

\usepackage[english]{babel}

\usepackage[left=1.75cm, top=2.7cm, bottom=2.7cm,right=1.75cm]{geometry}

\usepackage{verbatim}
\usepackage{mathrsfs}
\usepackage{bm}
\usepackage{color}
\usepackage{graphicx}

\usepackage{url}

\usepackage{bm}

\usepackage{colortbl}

 \newtheorem{thm}{Theorem}[section]
 \newtheorem{cor}[thm]{Corollary}
 \newtheorem{lem}[thm]{Lemma}
 \newtheorem{prop}[thm]{Proposition}
 \theoremstyle{definition}
 
 \theoremstyle{remark}
 \newtheorem{rem}[thm]{Remark}
 \numberwithin{equation}{section}

    \def\bbq{{\mathbb Q}}

    \def\calF{{\mathcal F}}

    \def\calC{{\mathcal C}}
    
    \def\calL{{\mathcal L}}

    \def\calL{{\mathcal L}}
    
    \def\calI{{\mathcal I}}

    \def\bbr{{\mathbb R}}
    \def\bbe{{\mathbb E}}
    \def\bbp{{\mathbb P}}
    
    \def\bbn{{\mathbb N}}

    \def\a{{a}}
    \def\b{{b}}
    \def\c{{c}}

    \def\f{{f}}

    \def\s{{s}}
    \def\t{{t}}

    \def\z{{z}}

    \def\X{{X}}

  \definecolor{Red}{rgb}{0.00, 0.00, 0.00}
    \newcommand{\Red}{\color{Red}}
    \definecolor{DRed}{rgb}{0.00, 0.00, 0.00}
    \newcommand{\DRed}{\color{DRed}}
    \definecolor{Green}{rgb}{0.00, 0.00, 0.00}
    \newcommand{\Green}{\color{Green}}
    \definecolor{Blue}{rgb}{0.00, 0.00, 0.00}
    \newcommand{\Blue}{\color{Blue}}
    \definecolor{PaleGrey}{rgb}{0,0,0}
    \newcommand{\PaleGrey}{\color{PaleGrey}}

\title{Small-time expansions of the distributions, densities, and option prices of stochastic volatility models with L\'evy jumps}
\author{Jos\'e E. Figueroa-L\'{o}pez\thanks{Department of Statistics, Purdue University, West Lafayette, IN, 47907,  USA ({\tt figueroa@purdue.edu}). {Research supported in part by the NSF grant: DMS-0906919.}}
 \and Ruoting Gong\thanks{School of Mathematics, Georgia Institute of Technology, Atlanta, GA, 30332, USA ({\tt rgong@math.gatech.edu}).}
 \and Christian Houdr\'e\thanks{School of Mathematics, Georgia Institute of Technology, Atlanta, GA, 30332, USA ({\tt houdre@math.gatech.edu}).}
}

\date{}
\begin{document}
\maketitle
\begin{abstract}
We consider a stochastic volatility model with L\'evy
jumps for a log-return process $Z=(Z_{t})_{t\geq 0}$ of the form $Z=U+X$,
where $U=(U_{t})_{t\geq 0}$ is a classical stochastic volatility process
and $X=(X_{t})_{t\geq 0}$ is an independent L\'evy process with absolutely
continuous L\'evy measure $\nu$. Small-time expansions, of arbitrary
polynomial order, in time-$t$, are obtained for the tails
$\bbp\left(Z_{t}\geq z\right)$, $z>0$, and for the call-option
prices  $\bbe\left(e^{z+Z_{t}}-1\right)_{+}$, $z\neq 0$,
assuming smoothness conditions on the {\PaleGrey density of $\nu$} away from the
origin and a small-time large deviation principle on $U$.
Our approach allows for a unified treatment of general payoff functions
of the form $\varphi(x){\bf 1}_{x\geq{}z}$ for smooth functions $\varphi$ and
$z>0$. As a consequence of our tail expansions, the polynomial expansions in $t$ of the transition densities $f_{t}$ are also {\Green obtained} under mild conditions.

\vspace{0.3 cm}


\vspace{0.3 cm}


\end{abstract}

\section{Introduction}

It is generally recognized that accurate modeling of the option market and asset prices requires a mixture of a continuous diffusive component and a jump component. For instance, based on high-frequency statistical methods for It\^o semimartingales, several empirical studies have {\PaleGrey statistically} rejected the null hypothesis of either a purely-jump or a purely-continuous model (see, e.g.,  \cite{AS02}, \cite{AitJacod06}, \cite{AitJacod10}, \cite{BNSh06}, \cite{Podolskij06}). Similarly, based on partially heuristic arguments, \cite{CW03} characterizes the small-time behavior of at-the-money (ATM) and out-of-the-money (OTM) call option prices, and based on these results, then argues that both, a continuous and a jump component, are {\PaleGrey needed} to explain the {behavior of the market implied volatilities for S\&P500 index options. {\Red Another empirical work in the same vein is \cite{MedSca07}, where a small-time small-log-moneyness approximation for the implied volatility surface was {\Green studied in the case of} a local jump-diffusion model with finite jump activity. Based on S\&P500 option market data, \cite{MedSca07} also concludes that jumps {\Blue are significant} in the risk-neutral price dynamics of the underlying asset}. 

Historically, local volatility models (and more recently stochastic volatility models) were the models of choice to replicate the skewness of the market implied volatilities at a given time (see \cite{Fouque} and \cite{Gatheral} for more details). However, it is a well-known empirical fact that implied volatility skewness is more accentuated as the expiration time approaches. Such a phenomenon is hard to reproduce within the purely-continuous framework unless the ``volatility of volatility" is forced to take very high values. Furthermore, as nicely explained in \cite{CT04} (Chapter 1), the very existence of a market for short-term options is evidence that option market participants operate under the assumption that a jump component is present.

In recent years the literature of small-time asymptotics for vanilla option prices of jump-diffusion models has grown significantly with strong emphasis to consider either a purely-continuous model or a purely L\'evy model. In {\PaleGrey case} of stochastic volatility models and local volatility models, we can mention, among others, \cite{Busca1:2002}, \cite{Busca2:2004}, \cite{FengFordeFouque}, \cite{Forde2:2010}, \cite{Forde:2010}, \cite{Gatheral:2009}, \cite{Henry:2009}, \cite{Paulot:2009}. For L\'evy {\PaleGrey processes}, \cite{Rop10} and  \cite{Tankov} show that OTM option prices are generally\footnote{That is, except for some pathological cases (see \cite{Rop10} for examples)} asymptotically equivalent to the time-to-maturity $\tau$ as $\tau\to{}0$. In turn, such a behavior implies that the implied volatilities of a L\'evy model explodes as $\tau\to{}0$. The exact first order asymptotic behavior of the implied volatility close to maturity was independently obtained in \cite{FigForde:2010} and \cite{Tankov}, while the former paper also gives the second order asymptotic behavior. There are very few articles that consider simultaneously stochastic volatility and jumps in the model. One
such work is \cite{CW03} which obtains, partially via heuristic
arguments, the first order asymptotic behavior of an It\^o semimartingale with jumps. Concretely, ATM option prices of pure-jump models of bounded variation decrease at the rate  $O(\tau)$, while they are just  $O(\sqrt{\tau})$ under the presence of a Brownian component. By considering a stable pure-jump component, \cite{CW03} also shows that, in general, the rate could be $O(\tau^{\beta})$, for some $\beta\in(0,1)$. For OTM options, it is also argued that the first order behavior is $O(e^{-c/\tau})$ in the purely-continuous case, while it becomes $O(\tau)$ under the presence of jumps. Very recently, \cite{Muhle} formally shows that the leading term of ATM option prices is of order $\sqrt{T}$ for a relatively general class of purely-continuous It\^o models, while for a more general type of It\^o processes with $\alpha$-stable-like small jumps, the leading term is $O(\tau^{1/\alpha})$ (see also \cite[Proposition 4.2]{FigForde:2010} and \cite[Proposition 5]{Tankov} for related results in exponential L\'evy models).
Fractional expansions are also obtained for the distribution functions of some L\'evy processes in \cite{Marchal}.

In this article, we consider a jump diffusion model by
combining a stochastic volatility model with a
pure-jump L\'evy process of infinite jump activity. More precisely, 
{\Red we consider a market consisting of a risk-free asset with constant interest rate $r\geq{}0$ and a risky asset with price process $(S_{t
})_{t\geq{}0}$ defined on a complete probability space $(\Omega,\calF,\bbq)$  equipped with a filtration
$(\calF_{t})_{t\geq{}0}$ satisfying the usual conditions. We assume that, under $\bbq$, the log-return process 
\begin{equation}\label{Eq:MIntZ}
	Z_{t}:=\log \frac{e^{-rt}S_{t}}{S_{0}},
\end{equation}
follows the jump diffusion model} 
\begin{eqnarray}\label{jumpdif}
    & Z_{t}=U_{t}+ X_{t},\quad
    dU_{t}= \mu(Y_{t}) dt +\sigma(Y_{t}) d W^{(1)}_{t},\\
    &d Y_{t}= \alpha (Y_{t}) dt +\gamma(Y_{t}) d
    W^{(2)}_{t}\label{jumpdifb}
\end{eqnarray}
with $U_{0}=X_{0}=0$, $Y_{0}=y_{0}\in\bbr$. Here, $W^{(1)}$ and
$W^{(2)}$ are Wiener processes adapted to $(\calF_{t})$, $X$ is an
independent $(\calF_{t})$-adapted pure-jump L\'{e}vy process with
triplet $(b,0,\nu)$, and $\sigma$, $\gamma$, $\mu$ and $\alpha$ are
deterministic functions such that (\ref{jumpdif})-(\ref{jumpdifb})
admits a solution. We also assume that $\int_{z>1} e^{z}\nu(dz)<\infty$ and
\begin{equation}\label{con30}
b=-\int_{\bbr}\left(e^{z}-1-z{\bf
1}_{|z|\leq{}1}\right)\nu(dz),\quad\text{and}\quad
\mu(y)=-\frac{1}{2}\sigma^{2}(y).
\end{equation}
{\Red In particular, note that $\bbq$ {\Blue is assumed to be} a martingale measure (i.e. $(e^{-rt}S_{t})_{t\geq{}0}$ is a
$\bbq$-martingale). The model (\ref{jumpdif})-(\ref{jumpdifb}) is appealing in practice since it incorporates jumps in the asset price process as well as volatility clustering and leverage effects.} {\DRed The process $(Y_{t})_{t\geq{}0}$ is the underlying volatility factor driving the stochastic volatility of the process.}

For $z\neq 0$ and $t>0$, let
\begin{equation}\label{KFP}
    G_{t}(z):=\bbe\left(e^{z+Z_{t}}-1\right)_{+},
\end{equation}
where $\bbe$ denotes, from now on, the expectation 
{\Red under $\bbq$}.
We will show that, under mild conditions, the following
small-time expansions for $G_{t}$ hold true:
\begin{align}\label{FExp1} G_{t}(z)
=\sum_{j=0}^{n}\b_{j}(z)\, \frac{t^{j}}{j!}
    +O(t^{n+1}),
\end{align}
for each $n\geq{}0$ and certain functions $\b_{j}$. Note that the
time-$t$ price of a European call option with strike $K$, which is
not at-the-money, {\DRed when the spot stock price and the underlying volatility factor have {\Green respective} values $s$ and $y_{0}$}, can then be expressed as
\begin{equation}
\label{Callopt}C(t,s):=e^{-r(T-t)}\bbe\left(\left.\left(S_{T}-K\right)_{+}\right|{\DRed \calF_{t}}\right)=Ke^{-r\tau}
G_{\tau}(\ln s-\ln K),
\end{equation}
where $\tau=T-t$ and $s\neq K$. Hence, the small-time behavior of (\ref{KFP}) leads to close-to-expiry approximations for the price of an arbitrary
not-at-the-money call option as a polynomial expansion in time-to-maturity. {\DRed From (\ref{Eq:MIntZ}), note also that the expression 
\begin{equation}\label{Eq:AntIntCOP}
	S_{0}e^{-z}G_{t}(z)=e^{-r t}\bbe\left(S_{t}-S_{0}e^{rt-z}\right)_{+}
\end{equation} 
can be interpreted as the time-$t$ call option price with log-moneyness\footnote{{\DRed As usual, log-moneyness is defined as the logarithm of the ratio between the strike $K$ and the forward price $S_{0}e^{rt}$.}} $\kappa:=-z$.}

{\Red Small-time option price asymptotics for the model {\Green (\ref{jumpdif})-(\ref{jumpdifb})} were also considered in \cite{MedSca07}, but only for finite-jump activity L\'evy processes $X$. Another important difference is that, here, we focus on small-time asymptotics for fixed log-moneyness $z\neq 0$}{\Green ,} while \cite{MedSca07} considers approximations where $z$ is simultaneously made to converge to $0$ as $t\to{}0$ ({\Red small-time and small-log-moneyness} asymptotics). Let us also {\Red remark} that \cite{MedSca07} assumes throughout (and without proof) that the implied volatility surface satisfies an expansion in powers of $z$ and $t$, which{\Green , in our opinion,} is a rather strong assumption. {\Red Another related work is \cite{Muhle}, where the first order {\Blue small-time small-log-moneyness} asymptotics {\Green is} considered for a class of It\^o semimartingales with non-zero continuous component (Theorem 3.1 therein). For the CGMY and related models}, higher order approximations for ATM option prices are {\DRed obtained in} \cite{FigGongHou:2011b}.

Our method of proof is built on a type of iterated Dynkin formula of the form
\begin{equation}\label{IDynkin0}
    \bbe g(X_{t})=g(0)
    +\sum_{k=1}^{n} \frac{t^{k}}{k!}
    L^{k}g(0)+
    \frac{t^{n+1}}{n!}
    \int_{0}^{1} (1-\alpha)^{n}\bbe\left(
    L^{n+1} g(X_{\alpha t})\right)d\alpha,
\end{equation}
where $g$ is a sufficiently smooth function and $L$ is the infinitesimal generator of the L\'evy process $X$ {\Red given by 
\begin{equation}\label{InfGenLevy0}
	Lg(x):=\frac{\sigma^{2}}{2}g''(x)+b g'(x)+\int \left(g(z+x)-g(x)-zg'(x){\bf 1}_{\{|z|\leq{}1\}}\right) \nu(dz),
\end{equation}
for $g\in C^{2}_{b}$ and a L\'evy triplet $(b,\sigma^{2},\nu)$ (see Section \ref{Sect:Not} below for {\Green terminology}).}
The main complication with option call prices arises from the lack of smoothness of the payoff function $g_{z}(x)=(e^{z+x}-1)_{+}$. In order to ``regularize" the payoff function $g$, we follow a two step procedure. First, we decompose the L\'evy process into a compound Poisson process with a smooth jump density vanishing in a neighborhood of the origin and an independent L\'evy process with small jumps. Then, we condition {\PaleGrey $\bbe g(X_{t})$} on the number of jumps of the compound Poisson component of $X$ and apply {\PaleGrey the} Dynkin's formula on each of the resulting terms. Contrary to the approaches in \cite{FigForde:2010} and \cite{Tankov}, where the special form of the payoff function $g_{z}(x)=(e^{z+x}-1)_{+}$ plays a key role, our approach can handle more general payoff functions of the form
\begin{equation}\label{GPFTC}
    g_{z}(x)=\varphi(x) {\bf 1}_{\{x\geq{}z\}},
\end{equation}
for a smooth function $\varphi$ and positive $z$. By considering the particular case $\varphi(x)\equiv 1$, we generalize the distribution expansions in \cite{FigHou:2008} to our jump-diffusion setting.
Let us emphasize that the process $Z$ in (\ref{jumpdif}) is not truly a Markov model but rather a hidden Markov model. This fact causes some technical subtleties that require a careful analysis of the iterated infinitesimal generator of  the bivariate Markov process $\{(U_{t},Y_{t})\}_{t\geq{}0}$.

As an equally relevant second contribution of our paper, we also obtain polynomial expansions for the transition densities
$f_{t}$ of the L\'evy process, under conditions
involving only the L\'evy
density of $X$. This is an important improvement to our former
results in \cite{FigHou:2008}, where a uniform boundedness
condition on all the derivatives of $f_{t}$ away from the origin was imposed.
Expansions for the transition densities of local volatility models (with possibly jumps but only of finite activity) have appeared before in the literature (e.g. see \cite{AitSahalia99}, \cite{AitSahalia08}, \cite{Yu}). Unlike our approach, the general idea in these papers consists of first proposing the general form of the expansion, and then choosing the coefficients so that either the backward or forward Kolmogorov equation is satisfied. The resulting coefficients typically involve iterated infinitesimal generators as in our expansions, even though our approximations are uniform away from the origin.

The paper is organized as follows. Section 2 contains some
preliminary results on L\'evy processes, which will be needed
throughout the paper. Section 3 establishes the small-time
expansions, of arbitrary polynomial order in $t$, for both the tail
distributions $\bbp(Z_{t}\geq z)$, $z>0$, and the call-option price
function $G_{t}(z)$, $z\neq 0$.  This section also justifies the validity of our results for payoff functions of the form (\ref{GPFTC}). Section 4 illustrates our approach by presenting the first few terms of those expansions. Interestingly enough, the first two
coefficients of the expansion of the general model coincide with the
first two coefficients of an exponential L\'evy model. Section 5
obtains the asymptotic behavior of the corresponding implied
volatility. Section 6 gives a small-time expansion for the transition density of a general
L\'evy process under rather mild conditions. The proofs of our main results are deferred to Appendices.

\section{Background and preliminary results}

\subsection{Notation}\label{Sect:Not}\hfill

Throughout this paper, $C^{n}$ {(or $C^{n}(\bbr)$)}, $n\geq 0$, is
the class of real valued functions, defined on $\bbr$, which have
continuous derivatives of order $0\leq k\leq n$, while $C_{b}^{n}\subset
C^{n}$ corresponds to those functions having bounded derivatives. In a
similar fashion, $C^{\infty}$ (or $C^{\infty}(\bbr)$) is the class of
real valued {\PaleGrey functions}, defined on $\bbr$, which have continuous
derivatives of any order $k\geq 0$, while
$C_{b}^{\infty}(\bbr)\subset C^{\infty}$ are again the ones having
bounded derivatives. Sometimes $\bbr$ will be replaced by
{\DRed $\bbr_{0}:=\bbr\setminus\{0\}$} or $\bbr^{k}$ when the functions are defined on
these spaces.

Throughout this section, $\X$ {\Red denotes} a L\'evy process with
triplet $(b,\sigma^{2},\nu)$ defined on
$\left(\Omega,\calF,(\calF_{t})_{t\geq{}0},\bbq\right)$. Let us
write $X$ in terms of its L\'evy-It\^o decomposition:
\[
    \X_{t}=bt+\sigma W_{t}+ \int_{0}^{t}\int_{|\z|>1} z \mu(ds,dz)+\int_{0}^{t}\int_{|z|\leq{}1} z\bar\mu(ds,dz),
\]
where $W$ is a Wiener process and $\mu$ is an independent Poisson
measure on $\bbr_{+}\times\bbr\backslash\{0\}$ with mean measure
$dt\nu(dz)$ and compensator $\bar\mu$.  For each $\varepsilon>0$,
let $c_{\varepsilon}\in C^{\infty}$ be a symmetric truncation
function such that ${\bf 1}_{[-\varepsilon/2,\varepsilon/2]}(z) \leq
c_{\varepsilon}(z) \leq {\bf 1}_{[-\varepsilon,\varepsilon]}(z)$ and {\Green let $\bar{c}_{\varepsilon}:=1-c_{\varepsilon}$. Next, for $0<\varepsilon<1$, consider two independent L\'evy processes, denoted by $\bar{X}^{\varepsilon}$ and ${X}^{\varepsilon}$, with respective L\'evy {\PaleGrey triplets} $(0,0,\bar{c}_{\varepsilon}(z)\nu(dz))$ and $(b_{\varepsilon},\sigma^{2},c_{\varepsilon}(z)\nu(dz))$, where 
\[
	 b_{\varepsilon}:=b-\int_{|z|\leq{}1}z \bar{c}_{\varepsilon}(z)\nu(dz).
\]
Note that $(X_{t})_{t\geq{}0}$ has the same law as $(X^{\varepsilon}_{t}+\bar{X}^{\varepsilon}_{t})_{t\geq{}0}$ and, therefore, without loss of generality, we assume hereafter that $X=X^{\varepsilon}+\bar{X}^{\varepsilon}$. Note also that}
$\bar{X}^{\varepsilon}$ is a
compound Poisson process with intensity of jumps
$\lambda_{\varepsilon}:=\int \bar{c}_{\varepsilon}(z)\nu(dz)$
and jumps distribution
$ \bar{c}_{\varepsilon}(z)\nu(dz)/\lambda_{\varepsilon}$. Throughout, $(\xi_{i}^{\varepsilon})_{i\geq{}1}$ stands for the jumps of the process $\bar{X}^{\varepsilon}$.
{\Green The} process  $X^{\varepsilon}$ has infinitesimal generator $L_{\varepsilon}$ given by
\begin{equation}\label{IG}
    {\DRed L_{\varepsilon} g(y)=
    b_{\varepsilon} g'(y)+\frac{\sigma^{2}}{2}g''(y)+\calI_{\varepsilon}g(y)},
\end{equation}
for $g\in C_{b}^{2}$, where {\DRed 
\begin{align*}
    \calI_{\varepsilon}g(y)&:=\int_{\bbr_{0}}
    \left\{g(y+z)-g(y)-zg'(y){\bf 1}_{|z|\leq{}1}\right\} c_{\varepsilon}(z)\nu(dz)\\
    &\,=\int_{\bbr_{0}}\int_{0}^{1} g''(y+\beta w)(1-\beta)d\beta w^{2}c_{\varepsilon}(w)\nu(dw).
\end{align*}
The} following tail estimate for $X^{\varepsilon}$ {\Red is also} used in
the sequel:
\begin{eqnarray}
\label{TailEstLevy} \bbp(|X_{t}^{\varepsilon}|\geq z)\leq
t^{az}\exp(az_{0}\ln z_{0})\exp(az-az\ln z),
\end{eqnarray}
where $a\in(0,\varepsilon^{-1})$, and $t,z>0$ satisfy $t<z/z_{0}$
for some $z_{0}$ depending only on $a$ (see \cite[Section
2.6]{Sato:1999}, \cite[Lemma 3.2]{Ruschendorf} and \cite[Remark
3.1]{FigHou:2008} for proofs and extensions).

Throughout the paper, we also make the following standing
assumptions:
\begin{align} \label{conB2}
    &\nu(dz)\!=\!s(z)dz,\,s\!\in\! C^{\infty}\!\left(\bbr\!\setminus\!\{0\}\right)\text{ and }\gamma_{k,\delta}\!:=\!\sup_{|z|>\delta}|s^{(k)}(z)|\!<\!\infty,  \text{ for all } \delta>0,\\
\label{conB1}&\int_{|z|>1}e^{c|z|}\nu(dz)<\infty,\text{ for some } c>2.
\end{align}
Finally, the following {\PaleGrey notation are also in use:}
\begin{align*}
    &{s}_{\varepsilon}:=c_{\varepsilon} s,\quad
    \bar{s}_{\varepsilon}:=(1- c_{\varepsilon})s, \quad
    L^{0} g =g,\quad L^{k+1} g =L(L^{k}g), \quad (k\geq{}0), \\
   &\bar{s}_{\varepsilon}^{*0}*g = g, \quad \bar{s}_{\varepsilon}^{*1}=\bar{s}_{\varepsilon}, \quad  \bar{s}_{\varepsilon}^{*k}(x)=\int  \bar{s}_{\varepsilon}^{*(k-1)}(x-u)\bar{s}_{\varepsilon}(u)du,\quad(k\geq 2).
\end{align*}

\subsection{Dynkin's formula for smooth subexponential functions}\hfill

Let us recall that taking expectations in the well-known Dynkin {\PaleGrey formula} gives:
\begin{equation}\label{Dynkin}
    \bbe g(X_{t}) =g(0)+\int_{0}^{t} \bbe\left(Lg(X_{u})\right)d u
    =g(0)+t\int_{0}^{1} \bbe\left(Lg(X_{\alpha t})\right)d\alpha,
\end{equation}
valid if $g\in C^{2}_{b}$. Iterating (\ref{Dynkin}), one obtains the
following expansion for $g\in C_{b}^{2n+2}$ (e.g.,
see~\cite[Proposition 9]{HPAS:1998}):
\begin{equation}\label{IDynkin}
    \bbe g(X_{t})=g(0)
    +\sum_{k=1}^{n} \frac{t^{k}}{k!}
    L^{k}g(0)+
    \frac{t^{n+1}}{n!}
    \int_{0}^{1} (1-\alpha)^{n}\bbe\left(
    L^{n+1} g(X_{\alpha t})\right)d\alpha.
\end{equation}
{\DRed Expansions of the form (\ref{IDynkin}) are then called iterated-type Dynkin expansions.} For our purposes, it will be useful to extend (\ref{IDynkin}) to
subexponential functions. The proof of the following result can be found in Appendix \ref{ApA}.
\begin{lem}\label{prop:IDynkin}
Let $\nu$ satisfy (\ref{conB1}), and let $g\in C^{2n+2}$ be such
that
\begin{equation}\label{NCID}
    \limsup_{|y|\to\infty} e^{-\frac{c}{2}|y|}|g^{(i)}(y)|<\infty,
\end{equation}
for any $0\leq i\leq 2n+2$. Then, (\ref{IDynkin}) holds true.
\end{lem}

Below, let 
\begin{align*}
    & b_{0}:=-\int_{\bbr}\bar{c}_{\varepsilon}(u) \nu(du),
    \quad
    {b_{1}:=b-\int_{\bbr}u ({c}_{\varepsilon}(u)-{\bf 1}_{|u|\leq{}1})\nu(du)},\\  &b_{2}:=\sigma^{2}/2, \quad{b_{3}:=\frac{1}{2}\int_{\bbr} u^{2} c_{\varepsilon}(u) \nu(du), \;\text{ and }\;
     b_{4}:= \int_{\bbr} \bar{c}_{\varepsilon}(u) \nu(du)}.
\end{align*}
Note that all these constants depend on $\varepsilon>0$, but
this is not explicitly indicated for the ease of notation.
In order to work with the iterated infinitesimal generator $L^{k}$
appearing in (\ref{IDynkin}), the forthcoming representation will
turn out to be useful (see \cite[Lemma 4.1]{FigHou:2008} for its
verification\footnote{{\Green For convenience we switched} the role
of $b_{3}$ and $b_{4}$.}).
\begin{lem}\label{RIGen}
    Let
    ${\bf \mathcal{K}}_{k}=\{{\bf k}=(k_{0},\dots,k_{4})\in\bbn^{5}:
    k_{0}+\dots+k_{4}=k\}$ and for ${\bf k}\in\mathcal{K}_{k}$,
    let $\ell_{\bf k}:=k_{1}+2k_{2}+2k_{3}$. Then, for any $k\geq{}1$ and $\varepsilon>0$,
\begin{equation}\label{DcmL2}
    L^{k}g (x)=\sum_{{\bf k}\in\mathcal{K}_{k}} \prod_{i=0}^{4} b_{i}^{k_{i}}
        \binom{k}{\bf k}B_{_{{\bf k},\varepsilon}}g (x),
\end{equation}
where
\begin{align*}
    B_{_{{\bf k},\varepsilon}}g(x)
    :=\left\{
    \begin{array}{ll}\int g^{({\ell_{\bf k}})}
    \left(x+\displaystyle{\sum_{j=1}^{k_{3}}}\beta_{j}w_{j}
        +\displaystyle{
    \sum_{i=1}^{k_{4}}}u_{i}
    \right)d\pi_{_{{\bf k},\varepsilon}},&\text{ if }k_{3}+k_{4}>0,\\
    g^{({\ell_{\bf k}})}
    \left(x\right), &\text{ if }k_{3}=k_{4}=0,\end{array}
    \right.
\end{align*}
and {\Green where} the above integral is with respect to the probability measure
\begin{align*}
    d\pi_{_{{\bf k},\varepsilon}}
    =\prod_{j=1}^{k_{3}}\frac{1}{b_{3}} c_{\varepsilon}(w_{j})w_{j}^{2}\nu(dw_{j})
(1-\beta_{j})d\beta_{j}\prod_{i=1}^{k_{4}}\frac{1}{b_{4}} \bar{c}_{\varepsilon}(u_{i})\nu(du_{i}),
\end{align*}
on $\bbr^{k_{3}}\times [0,1]^{k_{3}}\times \bbr^{k_{4}}$ (under the standard  conventions that $0/0=1$ and $\prod_{i=1}^{0}=1$).
\end{lem}
\begin{rem}
   The expansion (\ref{DcmL2}) holds true for (possibly unbounded) functions $g\in C^{2k+2}$ satisfying (\ref{NCID}) for any $0\leq{}i\leq{}2k+2$.
\end{rem}

\section{Small-time expansions for the tail distributions and option prices}

In this section, we derive the small-time expansions for both the
tail distribution $\bbp\left(Z_{t}\geq z\right)$, $z>0$, and for the
call-option price function $\bbe\left(e^{z+Z_{t}}-1\right)_{+}$,
$z\neq 0$. With an approach similar to that in \cite[Theorem
3.2]{FigHou:2008}, the idea is to apply the following general moment
expansion (easily obtained by conditioning on the number of jumps of
the process $\bar{X}_{t}^{\varepsilon}$ introduced {\Green in Section \ref{Sect:Not}}{\PaleGrey )}:
\begin{align}\label{T1b}
    \bbe f(Z_{t})&
    = e^{-\lambda_{\varepsilon}t}\bbe f\left(U_{t}+X^{\varepsilon}_{t}\right)
   +
    e^{-\lambda_{\varepsilon}t}\sum_{k=n+1}^{\infty}\frac{(\lambda_{\varepsilon}t)^{k}}{k!}
    \bbe f\left(U_{t}+X^{\varepsilon}_{t}+\sum_{i=1}^{k}\xi_{i}^{\varepsilon}\right)\\
    \label{T2b}
    &\quad+
    e^{-\lambda_{\varepsilon}t}\sum_{k=1}^{n}\frac{(\lambda_{\varepsilon}t)^{k}}{k!}
    \bbe f\left(U_{t}+X^{\varepsilon}_{t}+\sum_{i=1}^{k}\xi_{i}^{\varepsilon}\right),
\end{align}
where {\DRed $\{\xi^{\varepsilon}_{i}\}_{i\geq{}1}$} are the jumps of the process
$\bar{X}^{\varepsilon}$. We shall take $f(u)=f_{z}(u):={\bf
1}_{\{u\geq z\}}$ in order to obtain the expansion of the transition
distribution and $f(u)=f_{z}(u):=(e^{z+u}-1)_{+}$ in order to obtain
the expansion of the call-option price. To work out the terms in
(\ref{T2b}), we use the iterated formula (\ref{IDynkin}), while to
estimate the terms in (\ref{T1b}), we assume that the underlying
stochastic volatility model $U$ satisfies a {\DRed short-time ``large
deviation principle" of the form}:
\begin{eqnarray}\label{tailU0}
    \lim_{t\rightarrow 0}t\ln\bbp(U_{t}>u)=-\frac{1}{2}d(u)^{2},\quad (u>0),
\end{eqnarray}
where $d(u)$ is a strictly positive rate function. In Section
\ref{CndLD} we review conditions for (\ref{tailU0}) to hold.

{\DRed The expansions provided in {\Green the sequel} will hold uniformly outside a neighborhood of the origin. Concretely, for fixed $z_{0}>0$ and $\varepsilon>0$, the term $O_{\epsilon,z_{0}}(t^{j})$ {\Green denotes a quantity, depending on $z$, $\varepsilon$, and $t$,} such that 
\begin{equation}\label{FIntSpecONt}
    \sup_{0<t\leq t_{0}}\sup_{|z|\geq{}z_{0}}t^{-j}|O_{\varepsilon,z_{0}}(t^{j})|<\infty,
\end{equation}
for some $t_{0}>0$, small enough, depending {\Green itself} on $\varepsilon$ and $z_{0}$.
}

\subsection{Expansions for the tail distributions}\hfill

We first treat the case $f_{z}(u):={\bf 1}_{\{u\geq z\}}$. The following expansion for the tail distributions of $Z$ (whose proof
can be found in Appendix \ref{ApB}) holds true{\PaleGrey .}
\begin{thm}\label{jumpdifdist}
Let $z_{0}>0$, $n\geq 1$, and $0<\varepsilon<z_{0}/(n+1)\wedge 1$.
Let the dynamics of $Z$ be given by (\ref{jumpdif}) and (\ref{jumpdifb}), and the
conditions (\ref{conB2})-(\ref{conB1}) and (\ref{tailU0}) be
satisfied. Then, there exists $t_{0}>0$ such that, for any $z\geq
z_{0}$ and $0<t<t_{0}$,
\begin{eqnarray}\label{jumdifexpan1}
\bbp(Z_{t}\geq z)=e^{-\lambda_{\varepsilon} t}\sum_{j=1}^{n}\widehat{A}_{j,t}(z)\frac{t^{j}}{j!}+O_{\varepsilon,
z_{0}}(t^{n+1}),
\end{eqnarray}
where
\begin{equation}
    \widehat{A}_{j,t}(z):=\sum_{k=1}^{j}
    \binom{j}{k}\bbe\left( \left(L_{\varepsilon}^{j-k} \widehat{f}_{k,z}\right)(U_{t})\right),\nonumber
\end{equation}
with $\widehat{f}_{k,z}(y):=\int_{z-y}^{\infty} \bar{s}^{*
k}_{\varepsilon} (u) du$.
\end{thm}

The expression (\ref{jumdifexpan1}) is not really satisfactory since
the coefficients $\widehat{A}_{j,t}$ are time-dependent and so the
asymptotic {\DRed behavior of the tail probability $\bbp(Z_{t}\geq{}z)$} as $t\to{}0$ {\DRed is} unclear. In order to obtain an
expansion of $\widehat{A}_{j,t}$, we can further obtain an iterated {\DRed Dynkin expansion for}  $\bbe g(U_{t},Y_{t})$ {\DRed and a suitable moment function $g$}. Indeed, assuming for simplicity
that $W^{(1)}$ and $W^{(2)}$ are independent, $(U,Y)$ is a Markov
process with infinitesimal generator
\begin{equation}\label{IGSV}
    \calL g(u,y) =\mu(y)\frac{\partial g}{\partial u}
    +\frac{\sigma^{2}(y)}{2}\frac{\partial^{2} g}{\partial u^{2}}+\alpha(y)\frac{\partial g}{\partial y}+\frac{\gamma^{2}(y)}{2}\frac{\partial^{2} g}{\partial y^{2}},
\end{equation}
for $g\in C^{2}_{b}$. It\^o's formula and induction imply that
\begin{equation}\label{IDynkinDiff}
    \bbe g(U_{t},Y_{t})=g(u_{0},y_{0})
    +\sum_{k=1}^{n} \frac{t^{k}}{k!}
    \calL^{k}g(u_{0},y_{0})+
    \frac{t^{n+1}}{n!}
    \int_{0}^{1} (1-\alpha)^{n}\bbe\left\{
    \calL^{n+1} g(U_{\alpha t},Y_{\alpha t})\right\}d\alpha,
\end{equation}
for any function $g$ such that $\calL^{k}g(u,y)$ is well-defined and
belongs to $C_{b}$ for {\DRed any} $0\leq k\leq{} 2n+2$. As in the case of the
infinitesimal generator of $X$, one can view the operator
{(\ref{IGSV})} as the sum of four operators. However, given that in
general those operators do not commute, it is not possible to write
a simple closed-form expression for $ \calL^{k}g(u_{0},y_{0})$ as is the case for $X$. Nevertheless, the following result gives a
recursive method to get such an expression when
$\mu(y)=-\sigma^{2}(y)/2$ and $g(u,y)=h(u)$ as {\DRed needed} here {\PaleGrey .}
\begin{prop}\label{IDynkinDif}
Let the dynamics of $U$ and $Y$ be given by (\ref{jumpdif}) and (\ref{jumpdifb}) with
independent $W^{(1)}$ and $W^{(2)}$ and with $C^{\infty}$
deterministic functions $\alpha$, $\sigma^{2}$, and $\gamma^{2}$. {\DRed For $h,\tilde{h}\in C^{2}(\bbr)$,} 
let
\begin{equation}\label{DOpNd0}
  {\Red \calL_{1}} h(u):= h''(u)-h'(u), \quad
 {\Red \calL}_{2}\tilde{h}(y):= \frac{\gamma^{2}(y)}{2}\tilde{h}''(y)+\alpha(y)\tilde{h}'(y).
\end{equation}
Then, {\DRed for $h\in C^{2n+2}$ and $g(u,y):=h(u)$}, the infinitesimal
generator (\ref{IGSV}) is such that
\[
  \calL^{k}g(u,y) =\sum_{j=0}^{k} B_{j}^{k}(y){\Red \calL}_{1}^{j} h(u),
\quad\text{ for }{\DRed 0\leq k\leq{}n},
\]
where $B_{j}^{k}(y)$ are defined iteratively as follows:
\begin{align*}
 & B_{0}^{0}(y)=1,\quad  B_{j}^{k}(y)=0, \quad \forall j\notin\{{\DRed 1},\dots,k\}, \quad {\DRed k\geq{}1},\\
 & B_{j}^{k}(y)={\Red \calL}_{2}B_{j}^{k-1}(y)+\frac{\sigma^{2}(y)}{2}B_{j-1}^{k-1}(y),
\quad {\DRed 1}\leq{}j\leq{}k, \quad k\geq{}1.
\end{align*}
\end{prop}
\begin{proof}
The proof is done by induction.
\end{proof}
Using the previous result, we can easily {\DRed find} conditions for the
iterated formula (\ref{IDynkinDiff}) to hold. To this end, let us
{\DRed introduce} the following class of functions:
\begin{align*}
    \mathcal{C}_{l}^{n}=\big\{p\in\! C^{n} &: |p^{(i)}(x)|\leq\mathcal{M}_{n}(1+|x|),\text{ for all }0\leq{}i\leq{}n,\text{ and }\\
    &\quad \text{ for some }\mathcal{M}_{n}\!<\!\infty\text{ independent of }x\big\}.
\end{align*}

\begin{cor}\label{IDynkinSV}
In addition to the conditions of Proposition \ref{IDynkinDif}, let
$\gamma\in \mathcal{C}_{l}^{0}$ and let $\alpha,
\sigma^{2},\gamma^{2}\in \calC_{l}^{k}$, for any $k\geq{}0$. Then,
(\ref{IDynkinDiff}) is satisfied for  $g(u,y):=h(u)$ whenever
$h\in C^{2n+2}_{b}$.
\end{cor}
\begin{proof}
Using It\^o's formula and induction, (\ref{IDynkinDiff}) will hold
provided that
    \[
        \int_{0}^{t} \frac{\partial \calL^{n}g(U_{s},Y_{s})}{\partial u} \sigma(Y_{s}) d W^{(1)}_{s}, \quad
        \text{and}\quad
        \int_{0}^{t} \frac{\partial \calL^{n}g(U_{s},Y_{s})}{\partial y} \gamma(Y_{s}) d W^{(2)}_{s},
    \]
    are true martingales. For this, it suffices that
    \[
        \bbe \int_{0}^{t} \left(\frac{\partial \calL^{n}g}{\partial u} \sigma\right)^{2} d s<\infty, \quad
        \text{and}\quad
        \bbe \int_{0}^{t} \left(\frac{\partial \calL^{n}g}{\partial y} \gamma\right)^{2} d s<\infty,\quad \forall t\geq{}0.
    \]
    Let us recall that since $\alpha$ and $\gamma$ belong to $\calC_{l}^{0}$,
    \begin{equation}\label{NEDf}
        \sup_{s\leq{}t}\bbe |Y_{s}|^{2m}<\infty,
    \end{equation}
    for any $t\geq{}0$ and $m\geq{}1$ (this is similar to \cite[Problem 5.3.15]{Karatzas:1988}).
    Hence, given the representation of Proposition \ref{IDynkinDif}, it suffices to show that for some constants $\mathcal{M}_{i}^{n}<\infty$ and non-negative integers $r_{i}^{n}$:
    \begin{equation}\label{NECoef}
        \left|(B_{j}^{n})^{(i)}(y)\right|\leq\mathcal{M}_{i}^{n}(1+ |y|)^{r_{i}^{n}},
    \end{equation}
    for any $i,n\geq{}0$ and $0\leq{}j\leq{}n$. This claim can again be shown by induction since, given that it is satisfied for $n-1$ and using the iterative representation for $B_{j}^{n}$ in Proposition \ref{IDynkinDif},
    \begin{align*}
        \left|(B_{j}^{n})^{(i)}(y)\right|\leq \sum_{\ell=0}^{i} \binom{i}{\ell}&\left|\frac{1}{2}\left(\gamma^{2}\right)^{(\ell)}\left(B_{j}^{n-1}\right)^{(i-\ell+2)} \right.\\
        &\left.+\left(\alpha\right)^{(\ell)}\left(B_{j}^{n-1}\right)^{(i-\ell+1)}+\frac{1}{2}\left(\sigma^{2}\right)^{(\ell)}\left(B_{j-1}^{n-1}\right)^{(i-\ell)}\right|,
    \end{align*}
    which can be bounded by $\mathcal{M}_{i}^{n}(1+|y|)^{r_{i}^{n}}$ since, by assumption, $\sigma^{2}$, $\alpha$, and $\gamma^{2}$ belong to $\calC_{l}^{k}$, for all $k\geq{}0$.
\end{proof}
{\Green The} previous result covers the {\DRed Heston model,
\begin{equation}\label{Heston}
    dU_{t}= -\frac{1}{2}Y_{t} dt +\sqrt{Y_{t}} d W^{(1)}_{t},\quad
    d Y_{t}=\chi (\theta- Y_{t}) dt +v \sqrt{Y_{t}} d W^{(2)}_{t},
\end{equation}
as well as the exponential Ornstein-Uhlenbeck (OU) model:
\begin{equation}\label{OUExp}
    dU_{t}= -\frac{1}{2}e^{2Y_{t}} dt +e^{Y_{t}} d W^{(1)}_{t},\quad
    d Y_{t}= \chi (\theta- Y_{t}) dt +v d W^{(2)}_{t}.
\end{equation}
{\Green Both models are used} in practice.}

Let us now use Corollary \ref{IDynkinSV} to obtain a second order
expansion for $\bbe h(U_{t})$. Omitting, for the ease of notation,
the evaluation of the functions $B_{j}^{k}$ at $y_{0}$, we can write
\begin{align*}
    \bbe h(U_{t})&=h(0)+B_{1}^{1}{\Red \calL}_{1}h(0)t+
    \left(B_{1}^{2}{\Red \calL}_{1}h(0)+B_{2}^{2}{\Red \calL}_{1}^{2}h(0)\right) t^{2}+
    O(t^{3}),
\end{align*}
where
\begin{align}\label{FTT}
    B_{1}^{1}=\frac{1}{2}\sigma_{0}^{2}, \quad B_{1}^{2}= {\DRed \frac{\gamma_{0}^{2}}{2}\left(\sigma_{0}\sigma''_{0}+
(\sigma'_{0})^{2}\right)+\alpha_{0}\sigma_{0}\sigma'_{0}},\quad
    B_{2}^{2}=\frac{1}{4}\sigma_{0}^{4}.
\end{align}
Above, we set $\sigma_{0}=\sigma(y_{0})$,
$\sigma'_{0}=\sigma'(y_{0})$, and $\sigma''_{0}=\sigma''(y_{0})$,
with similar notation for the other functions. {\DRed For the Heston model (\ref{Heston}),  the coefficients in (\ref{FTT}) are}
\[
	{\DRed B_{1}^{1}=\frac{y_{0}}{2}, \quad B_{2}^{2}=\frac{y_{0}^{2}}{2},\quad B_{1}^{2}=\frac{\chi(\theta-y_{0})}{2}}.
\] 
{\DRed Similarly, for the exponential OU model (\ref{OUExp}),} 
\[
	{\DRed B_{1}^{1}=\frac{e^{2y_{0}}}{2}, \quad B_{2}^{2}=\frac{e^{4y_{0}}}{4},\quad 
	B_{1}^{2}=e^{2y_{0}}\left(v^{2}+\chi(\theta-y_{0})\right).}
\] 

A general {\Green polynomial} expansion of transition distributions is
as follows{\PaleGrey .}
\begin{thm}\label{MthDst}
     With the notations and the conditions of Theorem \ref{jumpdifdist} and Corollary \ref{IDynkinSV},
    \begin{equation}\label{jumdifexpan3}
\bbp(Z_{t}\geq z)= e^{-\lambda_{\varepsilon} t}\sum_{j=1}^{n}\widehat{a}_{j}(z)\frac{t^{j}}{j!}+O_{\varepsilon,
z_{0}}(t^{n+1}),
\end{equation}
where
\begin{align*}
    \widehat{a}_{j}(z)&:=\sum_{p+q+r=j}
    \binom{j}{p,q,r}L_{\varepsilon}^{q}
     \left(\sum_{m=0}^{r}B_{m}^{r}(y_{0}){\Red \calL}_{1}^{m}
(\widehat{f}_{p,z})\right)
    (0)\nonumber,
\end{align*}
setting $\widehat{f}_{0,z}(y)\equiv 0$ and where the summation
is over all non-negative integers $p,q,r$.
\end{thm}
\begin{proof} It is enough to plug the expansion
(\ref{IDynkinDiff}) into the coefficients of the first summation in
(\ref{expbrevef}) (See the proof of Theorem \ref{jumpdifdist} in Appendix B) and rearrange terms using Proposition
\ref{IDynkinDif}. Note that the last integral in (\ref{IDynkinDiff})
is bounded for $\widetilde{f}_{k,z}\in C^{\infty}_{b}(\bbr)$ from the representation in Proposition \ref{IDynkinDif} and the
estimates (\ref{NEDf})-(\ref{NECoef}).
\end{proof}

As a way to illustrate the expansions, note that in the case of constant volatility ($\alpha(y)=\gamma(y)\equiv0$),
\[
    B_{k}^{k}(y)\equiv\left(\frac{\sigma_{0}^{2}}{2}\right)^{k}, \quad  B_{j}^{k}(y)\equiv 0,\quad \forall j\neq k, \quad k\geq{}0.
\]
Hence,
\begin{align*}
    \widehat{a}_{j}(z)&:=\sum_{p+q+r=j}
    \binom{j}{p,q,r}\left(\frac{\sigma_{0}^{2}}{2}\right)^{r}
    L_{\varepsilon}^{q}\left({\Red \calL}_{1}^{r}(\widehat{f}_{p,z})\right)
    (0).
\end{align*}

{\Red 
\begin{rem}\label{Simplify}
\smallskip
\noindent 
\rm{{(i)}} Clearly, the expansion (\ref{jumdifexpan3}) leads to an expansion {\Blue of} the form: 
\begin{equation}\label{jumdifexpan3RealPow}
\bbp(Z_{t}\geq z)= \sum_{j=1}^{n}\breve{a}_{j}(z)\frac{t^{j}}{j!}+O_{\varepsilon,
z_{0}}(t^{n+1}).
\end{equation}
Indeed, by expanding $e^{-\lambda_{\varepsilon}t}$ in (\ref{jumdifexpan3}), we have 
\begin{equation}\label{Cnst1}
        \breve{a}_{k}(z)=\sum_{j=1}^{k}  \binom{k}{j} \widehat{a}_{j}(z) (-\lambda_{\varepsilon})^{k-j}.
\end{equation}

\smallskip
\noindent 
\rm{{(ii)}} 
The coefficients $\breve{a}_{k}(z)$ in (\ref{jumdifexpan3RealPow}) are actually independent {\Green of} $\varepsilon$ (for $\varepsilon$ small enough) since they can be defined iteratively as limits of $ \bbp\left(X_{t}\geq{}y\right)$ as follows: 
\begin{align*}
	\breve{a}_{1}(z)=\lim_{t\rightarrow{}0}\frac{1}{t}\bbp\left( Z_{t}\geq{}z\right),\quad
	\frac{\breve{a}_{k}(z)}{k!}=\lim_{t\rightarrow{}0}
        \frac{1}{t^{k}}\left\{
        \bbp\left(Z_{t}\geq{}z\right)-
        \sum_{j=1}^{k-1} \breve{a}_{j}(z)\,\frac{t^{j}}{j!}
        \right\}.
\end{align*}
\end{rem}
}


\subsection{Expansions for the call option price}\hfill

For $z\neq 0$ and $t>0$, let
\begin{equation}\label{jumpdifOpt}
    G_{t}(z):=\bbe\left(e^{z+Z_{t}}-1\right)_{+},
\end{equation}
where $Z$ is the jump-diffusion process given by (\ref{jumpdif}) and (\ref{jumpdifb}). We
proceed to derive the small-time expansion of $G_{t}$ as
$t\downarrow 0$. We first consider the out-of-the-money case $z<0$
from which one can easily derive the in-the-money case $z>0$ via
put-call parity (see Corollary \ref{jumpdifinthemoney} below).
Throughout this section, we set
\[
    f(u)=f_{z}(u):=(e^{z+u}-1)_{+},
\]
and we also assume the following uniform boundedness condition:
there exists $0<M<\infty$, such that
\begin{eqnarray}
\label{con320}0<\sigma(y)\leq M.
\end{eqnarray}
\begin{rem}\label{IDynkinSV2} Under the uniform boundedness condition
(\ref{con320}), it is easy to see that $\bbe e^{cU_{t}}<\infty$, for
some $c>2$.  Then, a proof similar to that {\Red given} {\Green in} Corollary \ref{IDynkinSV}, using the
representation of Proposition \ref{IDynkinDif}, shows that
(\ref{IDynkinDiff}) is satisfied for $g(u,y):=h(u)$, whenever $h\in
C^{2n+2}$ is a subexponential function satisfying (\ref{NCID}).
\end{rem}
The next theorem gives an expansion for the
out-of-the-money call option prices {\DRed in {\Green terms} of the integro-differential {\PaleGrey operators} $L_{\varepsilon}$ and $\calL_{1}$ defined in (\ref{IG}) and (\ref{DOpNd0})} (its proof is given in Appendix
\ref{ApC}){\PaleGrey .}
\begin{thm}\label{jumpdifopt}
Let $z_{0}<0$, $n\geq 1$, and {\Green $0<\varepsilon<-z_{0}/2(n+1) \wedge 1$}.
Let the dynamics of $Z$ be given by (\ref{jumpdif}) and (\ref{jumpdifb}), and the
conditions of both Theorem \ref{jumpdifdist} and Corollary
\ref{IDynkinSV} as well as (\ref{con320}) be satisfied. Then there
exists $t_{0}>0$ such that, for any $0<t<t_{0}$ and $z<z_{0}$,
\begin{equation}\label{jumdifexpan2}
G_{t}(z)=e^{-\lambda_{\varepsilon}t}\sum_{j=1}^{n}\widehat{b}_{j}(z)\frac{t^{j}}{j!}+O_{\varepsilon,
z_{0}}(t^{n+1}),
\end{equation}
where
\begin{equation}\label{Eq:ExpHatbGen}
\widehat{b}_{j}(z):=\sum_{p+q+r=j}\binom{j}{p,q,r}L_{\varepsilon}^{q}\left(\sum_{m=0}^{r}B_{m}^{r}(y_{0}){\Red \calL}_{1}^{m}(\widehat{f}_{p,z})\right)(0),
\end{equation}
with $\widehat{f}_{0,z}(y)\equiv 0$, and
\begin{equation*}
\widehat{f}_{k,z}(y):=\int_{\bbr}
f_{z}(y+u)\bar{s}^{*k}_{\varepsilon}(u)du=\int_{\bbr}
\left(e^{z+y+u}-1\right)_{+}\bar{s}^{*k}_{\varepsilon}(u)du.
\end{equation*}
\end{thm}

\begin{rem}
   {\DRed By expanding  $e^{-\lambda_{\varepsilon}t}$ in (\ref{jumdifexpan2})}, {\Red one obtains (\ref{FExp1}) with the following coefficients}
\begin{equation}\label{Cnst1}
       {\Green \b_{k}(z):=\sum_{j=1}^{k}  \binom{k}{j} \widehat{b}_{j}(z) (-\lambda_{\varepsilon})^{k-j}}.
\end{equation}
{\Red As it was the case {\Green for} {\Blue the} expansions {\Green of} the tail probabilities (see Remark \ref{Simplify}-(ii) above), the coefficients $b_{k}(z)$ {\Blue can be defined iteratively as certain limits of $G_{t}(z)=\bbe \left(e^{z+Z_{t}}-1\right)_{+}$,} {\DRed as $t\to{}0$,} and, {\Green therefore}, they are independent {\Green of} $\varepsilon$ (for $\varepsilon$ small enough).}  
\end{rem}
To deal with the in-the-money case $z>0$, note that
\begin{align*}
\bbe\left(e^{z+Z_{t}}-1\right)_{+}&=\bbe(e^{z+Z_{t}}-1)+(e^{z+Z_{t}}-1)_{-}\nonumber\\
{}&=e^{z}-1+\bbe(e^{z+Z_{t}}-1)_{-}.\nonumber
\end{align*}
The expansion of $\bbe(e^{z+Z_{t}}-1)_{-}$ {\Red with $z>0$} is similar to
that of $\bbe\left(e^{z+Z_{t}}-1\right)_{+}$  {\Red with $z<0$. Therefore, we obtain the following result.}
\begin{cor}\label{jumpdifinthemoney}
Let $z_{0}>0$, $n\geq 1$, and {\PaleGrey $0<\varepsilon<z_{0}/2(n+1)\wedge 1$}.
Under conditions of Theorem \ref{jumpdifdist}, there exists a
$t_{0}>0$ such that, for any $0<t<t_{0}$, $z>z_{0}$,
\begin{align}
G_{t}(z)=e^{z}-1+e^{\lambda_{\varepsilon}t}{\sum_{m=1}^{n}\widetilde{b}_{m}(z)\frac{t^{m}}{m!}}+O_{\varepsilon,
z_{0}}(t^{n+1}),
\end{align}
where
\begin{equation}
\widetilde{b}_{j}(z):=\sum_{i+j+k=m}\binom{m}{i,j,k}L_{\varepsilon}^{i}\left(\sum_{l=0}^{i}B_{l}^{i}(y_{0}){\Red \calL}_{1}^{l}\widehat{g}_{k,z}\right)(0)\nonumber
\end{equation}
with
\begin{equation}
    \widehat{g}_{k,z}(y):=\int_{\bbr}\left(e^{y+z+u}-1\right)_{-}\bar{s}_{\varepsilon}^{*k}(u) du.\nonumber
\end{equation}
\end{cor}

\begin{rem}
	{\Red The results {\Green of} this section provide expansions for the option price of {\PaleGrey out-of-the-money (OTM)} and {\PaleGrey in-the-money (ITM)} options in non-negative integer powers of the time to maturity. However, as stated in the introduction, {\PaleGrey at-the-money (ATM)} option prices are known to be expandable as fractional powers, at least in the first leading term (see  \cite[Proposition 5]{Tankov}, \cite[Proposition 4.2]{FigForde:2010}, and \cite{Muhle}).  A natural question is then to try to understand the reasons for such {\Blue radically} different asymptotic behaviors. The key assumption that allows us to obtain non-negative integer powers is the \emph{smoothness} of the L\'evy density $s$ outside the origin. Concretely, if $s$ were discontinuous or were not differentiable at a fixed log-moneyness value $z$, we {\Green would} expect that the expansions for the tail probability $\mathbb{P}(Z_{t}\geq{}z)$ and for the OTM call option price $\mathbb{E} (e^{k+Z_{t}}-1)_{+}$ {\Green would} {\DRed exhibit} fractional powers. This fact was pointed out in \cite{Marchal} for the tail probabilities and a particular type of pure-L\'evy processes. The same {\Green is} true for $\mathbb{E} (e^{k+Z_{t}}-1)_{+}$, in view of its representation as the difference of two tail probabilities (see (14) in \cite{FigForde:2010}).  In {\Green view} of this observation, it is now {\PaleGrey intuitively} {\Green clear} {\PaleGrey that} ATM option prices typically exhibit fractional leading terms. Indeed, at-the-money option prices $\mathbb{E} (e^{Z_{t}}-1)_{+}$ {\Green correspond} to the log-moneyness $z=0$ and the L\'evy density $s$, is discontinuous at $0$ for an infinite-jump activity L\'evy process.} 
\end{rem}

\subsection{Other payoff functions}\hfill

One of the advantages of our approach is that it can be applied to more general payoff functions. Concretely, consider a function of the form:
\[
    f_{z}(u):=\varphi(u){\bf 1}_{\{u\geq{}z\}},
\]
where $\varphi\in C_{b}^{\infty}$. One can easily verify that, under the conditions of Theorem \ref{MthDst} {\Red and for $z>0$},
    \begin{equation}\label{jumdifexpan3B}
    \bbe f_{z}(Z_{t})= e^{-\lambda_{\varepsilon} t}\sum_{j=1}^{n}\widetilde{a}_{j}(z)\frac{t^{j}}{j!}+O_{\varepsilon,
z_{0}}(t^{n+1}),
\end{equation}
where
\begin{align*}
    \widetilde{a}_{j}(z)&:=\sum_{p+q+r=j}
    \binom{j}{p,q,r}L_{\varepsilon}^{q}
     \left(\sum_{m=0}^{r}B_{m}^{r}(y_{0}){\Red \calL}_{1}^{m}(\widehat{f}_{p,z})\right)
    (0)\nonumber,
\end{align*}
with $\widehat{f}_{0,z}(y)=0$, and
$\widehat{f}_{k,z}(y):=\int_{\bbr}
f_{z}(y+u)\bar{s}^{*k}_{\varepsilon}(u)du=\int_{z-y}^{\infty}\varphi(y+u)\bar{s}^{*k}_{\varepsilon}(u)du$.
Indeed, from the proof of Theorem  \ref{jumpdifdist} (which is the
key to Theorem \ref{MthDst}), the only step that requires some
extra care is to justify that
\[
    \widetilde{f}_{k,z}(y) :=\lambda_{\varepsilon}^{-k}\int_{z-y}^{\infty}\varphi(y+u) \bar{s}_{\varepsilon}^{*k}(u)du
\]
is in $C^{\infty}$ and that $\sup_{y}|\widetilde{f}_{k,z}^{(j)}(y)|<\infty$. This is proved by verifying ({\PaleGrey via} induction) that
\[
    \widetilde{f}_{k,z}^{(j)}(y)=\lambda_{\varepsilon}^{-k}\int_{z-y}^{\infty}\varphi^{(j)}(y+u) \bar{s}_{\varepsilon}^{*k}(u)du
    +\lambda_{\varepsilon}^{-k}\sum_{i=0}^{j-1}(-1)^{i}\varphi^{(j-1-i)}(z) \bar{s}_{\varepsilon}^{*(k-1)}*\bar{s}_{\varepsilon}^{(i)}(z-y).
\]

Similarly, under the stronger conditions of Theorem \ref{jumpdifopt}, one can easily consider payoff functions of the form
\[
    f_{z}(u):=\varphi(u){\bf 1}_{\{u\geq{}-z\}}, \quad(z<0),
\]
with $\varphi\in C^{\infty}$ such that $|\varphi^{(j)}(u)|\leq{}
M_{j} e^{u}$ for some constant $M_{j}<\infty$ and all $j\geq{}0$.

\subsection{On the {\DRed short-time} large deviation principle for diffusions}\label{CndLD} \hfill

Large deviation results of the form (\ref{tailU0}) have recently
been developed for different stochastic volatility (SV) models. For
instance, for uncorrelated SV models, \cite{Forde2:2010} shows (\ref{tailU0}) under the following
conditions:
\begin{align}
\label{con31}\bullet&\quad\text{The function }\alpha\text{ is {\PaleGrey uniformly} bounded and {\PaleGrey Lipschitz} continuous in } \bbr.\\
\label{con32}\bullet&\quad\exists\,M_{2}>M_{1}>0,\,\,s.t.\,\,0\leq
M_{1}\leq\sigma(y)\wedge\gamma(y)\leq\sigma(y)\vee\gamma(y)\leq
M_{2}<\infty.\\
\label{con33}\bullet&\quad\sigma,\,\gamma\in
C^{\infty},\text{ and }\sigma(y)\rightarrow\sigma_{\pm},\,\,\gamma(y)\rightarrow\gamma_{\pm},\text{ as }y\rightarrow\pm\infty.\\
\label{con34}\bullet &\quad\sigma\text{ and }\gamma\text{ are diffeomorphisms with }\sigma'>0\text{ and }\gamma'>0.\\
\label{con35}\bullet&\quad\exists\,y_{c}\in\bbr,\text{ such that }\sigma''>0,\,\,\gamma''>0\,\,for\,\,y<y_{c},\,\,\sigma''<0,\,\,\gamma''<0\\
&\quad\text{for }y>y_{c}\text{ and }\sigma'\vee\gamma'<M<\infty\text{ for some }M>0.\nonumber\\
\label{con36}\bullet&\quad\text{The function
}u\mapsto\frac{\gamma(\sigma^{-1}(u))}{u}\text{ is non-increasing}.
\end{align}
We refer to \cite{Forde2:2010} for an explicit expression for the
rate function $I$, which is not relevant here. The Heston model
(\ref{Heston}) (even with correlated Wiener processes $W^{(1)}$ and
$W^{(2)}$) was also considered in \cite{Forde:2009} and
\cite{Forde:2010}.

\section{{\DRed Particular examples and analysis of the parameter contributions}}
{\DRed In this section, we shed some light on the specific forms of the general expansions obtained in the previous section for some {\Green specific} models.
We also analyze the contribution, to the call option prices, of the various parameters of the model.} 

\subsection{{\DRed Exponential L\'evy models}}\hfill

{\DRed Let us start by considering the exponential L\'evy model}, which {\DRed can be seen as} a particular case of the
jump-diffusion models (\ref{jumpdif}) and (\ref{jumpdifb}) {\DRed by setting $\alpha(y)\equiv \gamma(y)\equiv 0$ and $\sigma(y)\equiv \sigma$. In this setting,} the log-return process $Z$ {\DRed is} a L\'evy process with {\PaleGrey generating} triplet $(b,\sigma^{2},\nu)$. Then, the following expansion for the {\PaleGrey out-of-the-money} call option price holds
true (see also \cite{FigForde:2010} {\DRed for an alternative derivation}). Note that the process $Z$ in this section is the same as the process $X$ in Section 2 and therefore all the notations are transferred accordingly.
\begin{cor}\label{OMoptLevy} Let $z_{0}<0$, $n\geq 1$, and
{\PaleGrey $0<\varepsilon<-z_{0}/2(n+1)\wedge 1$}. Let $Z=(Z_{t})_{t\geq 0}$ be a
L\'evy process with triplet $(b,\sigma^{2},\nu)$ satisfying
(\ref{conB2})-(\ref{conB1}). Then there exists $t_{0}>0$ such that,
for any $z< z_{0}$ and $0<t<t_{0}$,
\begin{align}\label{OMexLevy}
    G_{t}(z) =e^{-\lambda_{\varepsilon}t}\sum_{j=1}^{n}\c_{j}(z)\, \frac{t^{j}}{j!}
    +O_{\varepsilon,z_{0}}(t^{n+1}),
\end{align}
where
\[
    \c_{j}(z):=
    \sum_{k=1}^{j} \binom{j}{k} L_{\varepsilon}^{j-k}\widehat{h}_{k,z}(0),
\]
with
\begin{align*}
    \widehat{h}_{k,z}(y)
    :=\int_{\bbr} \left(e^{z+y+u}-1\right)_{+}\bar{s}_{\varepsilon}^{*k}(u) du.
\end{align*}
\end{cor}

{\Red As {\Green in} Corollary \ref{jumpdifinthemoney}, we also have the following result for the in-the-money case.}
\begin{cor}\label{IMoptLevy}
Let $z_{0}>0$, $n\geq 1$, and $0<\varepsilon<z_{0}/(n+1)\wedge{}1$.
Then, there exists $t_{0}>0$ such that, for
    any $z>z_{0}$ and $0<t<t_{0}$,
\begin{align}\label{IMexLevy}
    G_{t}(z) = e^{z}-1+e^{-\lambda_{\varepsilon}t}\sum_{j=1}^{n}\tilde{c}_{j}(z)\, \frac{t^{j}}{j!}
    +O_{\varepsilon,z_{0}}(t^{n+1}),
\end{align}
where
\[
    \tilde{c}_{j}(z):=
    \sum_{k=1}^{j}
    \binom{j}{k}L_{\varepsilon}^{j-k}\tilde{h}_{k,z}(0),\quad\tilde{h}_{k,z}(y)
    :=\int_{\bbr} \left(e^{z+y+u}-1\right)_{-}\bar{s}_{\varepsilon}^{*k}(u) du.
\]
\end{cor}

\subsection{{\DRed Second-order expansions}}\hfill

Here are the first two coefficients of (\ref{jumdifexpan3})  for
$\varepsilon>0$ small enough:
\begin{align*}
    \widehat{a}_{1}(z)&=B_{0}^{0}(y_{0}) \widehat{f}_{1,z}(0)=\int_{\bbr}f_{z}(u) \bar{s}_{\varepsilon}(u)du=\int_{z}^{\infty}{{s}(u)}du;\\
    \widehat{a}_{2}(z)&=\underbrace{2 L_{\varepsilon}( \widehat{f}_{1,z})(0)}_{p=1,q=1,r=0}
    +\underbrace{2 B_{1}^{1}(y_{0}) {\Red \calL}_{1}(\widehat{f}_{1,z})(0)}_{p=1,q=0,r=1}
    +\underbrace{\widehat{f}_{2,z}(0)}_{p=2,q=0,r=0}\\
    &=2\left(b_{\varepsilon}{{s}(z)}-\int_{\bbr}\int_{0}^{1}{{s}'(z-\beta u)}(1-\beta)d\beta u^{2} s_{\varepsilon}(u)du\right)\\
    &\quad -\sigma^{2}(y_{0}){({s}'(z)+{s}(z))}+\int_{\bbr^{2}}{\bf 1}_{\{u_{1}+u_{2}\geq{} z\}}\bar{s}_{\varepsilon}(u_{1})
    \bar{s}_{\varepsilon}(u_{2}) du_{1} du_{2}.
\end{align*}
The corresponding coefficients for (\ref{jumdifexpan2}) are obtained
as above with $f_{z}(x) ={\bf 1}_{\{x\geq{}z\}}$ replaced by
$f_{z}(y):=(e^{z+y}-1)_{+}$ with $z<0$. Hence, for
$\varepsilon>0$ small enough,
\begin{align}
    \widehat{b}_{1}(z)&=\int_{\bbr}\left(e^{z+u}-1\right)_{+}{{s}(u)}du;\nonumber
    \\
    \widehat{b}_{2}(z)&=\sigma^{2}(y_{0}){{s}(-z)}+2b_{\varepsilon}\!\int_{-z}^{\infty}\!e^{z+u}{{s}(u)}du \nonumber \\
    &\quad +\int_{\bbr^{2}}\!\!\left(e^{z+u_{1}+u_{2}}-1\right)_{+}\!\!\bar{s}_{\varepsilon}(u_{1})
    \bar{s}_{\varepsilon}(u_{2}) du_{1} du_{2} \label{Eq:hatb2Trm}\\
    &\quad+2\int_{\bbr}\int_{0}^{1} (1-\beta)\left(
    \int_{-z-\beta u}^{\infty}e^{z-\beta u+w}{{s}(w)}dw+{{s}(-z-\beta u)}\right)d\beta u^{2} s_{\varepsilon}(u)du.\nonumber
\end{align}
In the previous expressions one can substitute $c_{\varepsilon}(y)$
and $\bar{c}_{\varepsilon}(y)$ by ${\bf 1}_{0<|y|<\varepsilon}$ and
${\bf 1}_{|y|\geq \varepsilon}$, respectively.

\begin{rem}
{\DRed Combining (\ref{FExp1}), (\ref{Callopt}), and the expression for
$\hat{b}_{1}(z)=\b_{1}(z)$ above {\Green gives} the {\DRed following} expansion for the
price function of the {\PaleGrey out-of-the-money} call option near the expiration
$T$:}
\begin{align}
C(t,s)&=Ke^{-r(T-t)}G_{T-t}(\ln s-\ln
K)\nonumber\\
\label{OptExp1}&=(T-t)\int_{\bbr}\left(se^{u}-K\right)_{+}{{s}(u)}du+O_{\varepsilon,\ln(s/K)}\big((T-t)^{2}\big).
\end{align}
{\DRed In particular, {\Green the} leading term for out-of-the-money call option prices is ``dominated" by the jump component of the model regardless of the underlying continuous stochastic volatility component $U$.}
\end{rem}

\begin{rem}\label{CnnexLevy}
{\DRed Given that in (\ref{FTT}), $B_{0}^{0}(y_{0})$ and $B_{1}^{1}(y_{0})$
depend only on $\sigma(y_{0})$, it is interesting to note that the
first two coefficients in our expansions (\ref{jumdifexpan3}) and
(\ref{jumdifexpan2}) coincide with the coefficients corresponding to
an exponential L\'evy model with volatility $\sigma=\sigma(y_{0})$. In fact, the initial values
of $\alpha$ and $\gamma$ begin to appear with the {\PaleGrey coefficients}
$\widehat{a}_{3}(z)$ and $\widehat{b}_{3}(z)$ through the
coefficients $B_{1}^{2}$ in (\ref{FTT}).}
\end{rem}

\subsection{{\DRed Parameters contributions}}\hfill

{\DRed For simplicity, let $S_{0}=1$ and $r=0$. 
The {\DRed first term on the right-hand side of} (\ref{Eq:hatb2Trm}) {\DRed  indicates} the main contribution of the spot volatility $\sigma(y_{0})$ {\Green in} the call price. Indeed, recalling (\ref{Eq:AntIntCOP}) and the subsequent interpretation of $-z$ as the log-moneyness $\kappa$, one can say that the spot volatility has the effect of increasing the call price by 
\begin{equation}\label{Eq:CntrVol2ndTrm}
	\sigma^{2}(y_{0})e^{\kappa}s(\kappa)\frac{t^{2}}{2}\left(1+o(1)\right).
\end{equation}
This observation was first pointed out in \cite{FigForde:2010} for an exponential L\'evy model. 

As explained in Remark \ref{CnnexLevy}, the specific information {\Green present in}  the stochastic volatility model manifests itself up to {\PaleGrey an} $O(t^{3})$ term. Indeed, by analyzing each of the terms of the expansion for $\widehat{b}_{3}$ given in (\ref{Eq:ExpHatbGen}), the spot drift $\alpha(y_{0})$ and {\Green the} diffusion component $\gamma(y_{0})$ of the underlying stochastic volatility factor $Y$ {\Green are} only felt through the term $B_{1}^{2}(y_{0}) e^{\kappa}s(\kappa)t^{3}/2$, which, in light of (\ref{FTT}), takes the form:
\begin{equation}\label{Eq:SVLOEff}
	\left(\frac{\gamma^{2}(y_{0})}{2}\left(\sigma(y_{0})\sigma''(y_{0})+
(\sigma'(y_{0}))^{2}\right)+\alpha(y_{0})\sigma(y_{0})\sigma'(y_{0})\right)  e^{\kappa}s(\kappa)\, \frac{t^{3}}{2}.
\end{equation}
For instance, in the Heston model (\ref{Heston}), {\DRed $\sigma(y)=\sqrt{y}$ and} the effect of the stochastic volatility starts to {\DRed being felt} through the third order term 
\begin{equation}\label{Eq:EffChi}
	\frac{\chi(\theta -y_{0})}{4}e^{\kappa}s(\kappa)\,t^{3}.
\end{equation}
So, when the volatility is high {\DRed (low), relative to its average value $\theta$, a mean-reversion speed of $\chi$ will decrease (increase) the call price by the amount (\ref{Eq:EffChi})}. This result is intuitive as the larger the speed is, the faster one would expect the volatility to come back to its average value {\DRed when this is high, hence resulting in a discount}. In a Heston model, the effect of the volatility of volatility $v$ appears up to a fourth degree term of the form 
\[
	v^{2} y_{0}e^{\kappa}\left(s''(\kappa)-s'(\kappa)\right)\frac{t^{4}}{4!}.
\]
In contrast for the exponential OU model (\ref{OUExp}), the quantity (\ref{Eq:SVLOEff}) simplifies {\Green to}
\[
	\frac{e^{2 y_{0}}\left(\chi(\theta -y_{0})+v^{2}\right)}{2}e^{\kappa}s(\kappa)\, t^{3}
\]
and both the speed $\chi$ of mean reversion and the volatility  of volatility $v$ already appear in the third order term. {\DRed {\Green In} spite of appearing in a third-order term, in some situations the contribution of the stochastic volatility model could be significant such as under a fast-mean reverting framework where the value of $\chi$ is high. This {\Green phenomenon} has already been suggested in the literature, partially motivated by empirical evidence (see, e.g., \cite{Fouque}).}

By analyzing the terms in (\ref{Eq:ExpHatbGen}), one can see that in addition to the term (\ref{Eq:SVLOEff}) {\DRed (which {\Green clearly} involves the information of $\sigma$ at $y_{0}$)}, the contribution of the volatility $\sigma(\cdot)$ to the third-order term {\DRed $\widehat{b}_{3}$} also appears in the following term(s):
\begin{equation}\label{Eq:Cntrthirdtermssigm}
	\left(\frac{3}{4}\sigma^{4}(y_{0})\left(s''(\kappa)+s'(\kappa)\right)
	+\frac{3\sigma^{2}(y_{0})}{2}\left(\bar{s}^{*2}_{\varepsilon}(\kappa)-2b_{\varepsilon} s'(\kappa)\right)\right)e^{\kappa}\frac{t^{3}}{3!}.
\end{equation}

Identifying the effect of each model component at each level of resolution $t,t^{2},\dots$, might be more enlightening than just explicitly writing down each of the terms. For instance, one can say that under a Heston model with independent L\'evy jumps, the first and second-order terms are the same as those of the underlying exponential L\'evy model with volatility $\sigma(y_{0})=y_{0}$; while the third order term is the superposition of the corresponding third order term of the underlying  exponential L\'evy model with volatility $\sigma(y_{0})=y_{0}$ and the term (\ref{Eq:EffChi}). From here, we can write the third order term of the underlying exponential L\'evy model using the expansion of Theorem \ref{jumpdifopt} or using the following representation obtained in \cite{FigForde:2010} (see the proof of Proposition 2.2 therein):
\begin{equation}\label{Eq:FigFRep}
	\bbe\left(e^{\widetilde{X}_{t}}-e^{\kappa}\right)_{+}=\bbp({X}^{*}_{t}\geq{}\kappa)
	-e^{\kappa}\bbp(\widetilde{X}_{t}\geq{}\kappa).
\end{equation}
Here, $\widetilde{X}$ represents a L\'evy process with triplet $(\tilde{b},\sigma^{2}(y_{0}),\nu)$, while ${X}^{*}$ represents a L\'evy process with triplet $(b^{*},\sigma^{2}(y_{0}),\nu^{*})$, where 
\begin{align}
\label{Eq:ExprTildeb}
	&
	\tilde{b}=-\frac{1}{2}\sigma^2(y_{0})- \int_{-\infty}^{\infty} (e^x-1- x1_{|x|\le
1})\nu(dx),\quad \nu^{*}(dx)=e^{x}\nu(dx),\\
\label{Eq:ExprbStar}
&b^{*}=\frac{1}{2}\sigma^2(y_{0})- \int_{-\infty}^{\infty} (e^x-1- x1_{|x|\le
1})\nu(dx)+\int_{|x|\leq{}1} x \left(e^{x}-1\right)\nu(dx).
\end{align}
Next, let $d_{j}(\kappa;b,\sigma,\nu)$ be the $j^{th}$-order term in the tail expansion {\DRed of a L\'evy process $X$ with triplet $(b,\sigma^{2},\nu)$}, i.e., {\DRed the $d_{j}$'s are such that} 
\[
	\bbp(X_{t}\geq{}\kappa)=\sum_{j=1}^{n}d_{j}(\kappa; b,\sigma,\nu) \frac{t^{j}}{j!}+O(t^{n+1}).
\]
{\DRed Then}, the third-order term of the option price (\ref{Eq:FigFRep}) takes the form
\[
	\left(d_{3}(\kappa;b^{*},\sigma(y_{0}),\nu^{*})
	-e^{\kappa}d_{3}(\kappa;\tilde{b},\sigma(y_{0}),\nu)\right)\frac{t^{3}}{3!},
\]
which in turn {\Green allows} us to write the third-order term of the model (\ref{Eq:MIntZ})-(\ref{jumpdifb}) with the Heston specification (\ref{Heston}) as 
\[
	\left(d_{3}(\kappa;b^{*},\sigma(y_{0}),\nu^{*})
	-e^{\kappa}d_{3}(\kappa;\tilde{b},\sigma(y_{0}),\nu)
	+\frac{3\chi(\theta -y_{0})}{2}e^{\kappa}s(\kappa)\right)\frac{t^{3}}{3!}.
\]
We can further disentangle the volatility effect from the second and third order terms by using (\ref{Eq:CntrVol2ndTrm}) and (\ref{Eq:Cntrthirdtermssigm}). {\Green Specifically,} under the Heston model (where $\sigma(y)=\sqrt{y}$), the second-order term admits the representation
\begin{equation}\label{Eq:SecndOrdCrr}
	\left(d_{2}(\kappa;b^{*},0,\nu^{*})
	-e^{\kappa}d_{2}(\kappa;\tilde{b},0,\nu)+y_{0}e^{\kappa}s(\kappa)\right)\frac{t^{2}}{2},
\end{equation}
while the third order term can be written as 
\begin{align}\nonumber
	\frac{t^{3}}{3!}&\bigg(d_{3}(\kappa;b^{*},0,\nu^{*})
	-e^{\kappa}d_{3}(\kappa;\tilde{b},0,\nu)\\
	\nonumber
	&\quad
	+\frac{3\chi(\theta -y_{0})}{2}e^{\kappa}s(\kappa) +\frac{3}{4}y_{0}^{2}\left(s''(\kappa)+s'(\kappa)\right)e^{\kappa}\\
	\label{Eq:LstLnt3Trm}
	&\quad +\frac{3}{2}y_{0}\left(\bar{s}^{*2}_{\varepsilon}(\kappa)-2b_{\varepsilon} s'(\kappa)-2s(\kappa)\lambda_{\varepsilon}\right)e^{\kappa}\bigg).
\end{align}
In {\Green these} expressions, $b^{*}$ and $\tilde{b}$ are computed as in (\ref{Eq:ExprTildeb})-(\ref{Eq:ExprbStar}) but setting $\sigma(y_{0})=0$. {\DRed Note also that the term $-s(\kappa)e^{\kappa}\lambda_{\varepsilon}t^{3}/2$ in (\ref{Eq:LstLnt3Trm}) comes from expanding $e^{-\lambda_{\varepsilon}t}$ {\PaleGrey in} the second-order term $e^{-\lambda_{\varepsilon}t}\,\widehat{b}_{2}(z)t^{2}/2$, and isolating the contribution of $\sigma(y_{0})$ in the latter term (see (\ref{Eq:CntrVol2ndTrm})).}
The expressions for $d_{3}(\kappa;b^{*},0,\nu^{*})$ and  $d_{3}(\kappa;\tilde{b},0,\nu)$ can easily be written from the expressions in \cite{FigHou:2008} (see Theorem 4.3 therein) or from Theorem \ref{MthDst} above.} 

{\DRed Alternatively, one could also use (\ref{Eq:SecndOrdCrr}) and (\ref{Eq:LstLnt3Trm}) to develop {\PaleGrey the} second and third order SV corrections to call option prices of pure-jump L\'evy models. Concretely, the following expansions hold for the model (\ref{Eq:MIntZ})-(\ref{jumpdifb}) with the Heston SV specification (\ref{Heston}):
\begin{align*}
	\bbe\left(S_{t}-e^{\kappa}\right)_{+}
	&=\bbe\left(e^{X_{t}}-e^{\kappa}\right)_{+}+\sigma^{2}(y_{0})e^{\kappa}s(\kappa)\frac{t^{2}}{2}+O(t^{3}),\\		
	&=\bbe\left(e^{X_{t}}-e^{\kappa}\right)_{+}+\sigma^{2}(y_{0})e^{\kappa}s(\kappa)\frac{t^{2}}{2}\\
	&\quad +\bigg(\frac{3\chi(\theta -y_{0})}{2}e^{\kappa}s(\kappa) +\frac{3}{4}y_{0}^{2}\left(s''(\kappa)+s'(\kappa)\right)e^{\kappa}\\
	\label{Eq:LstLnt3Trm}
	&\quad \qquad +\frac{3}{2}y_{0}\left(\bar{s}^{*2}_{\varepsilon}(\kappa)-2b_{\varepsilon} s'(\kappa)-2s(\kappa)\lambda_{\varepsilon}\right)e^{\kappa}\bigg)\frac{t^{3}}{3!} +O(t^{4}).
\end{align*}
Above,  $X$ denotes the pure-jump L\'evy model underlying $S$, i.e., $X$ is a L\'evy process with L\'evy triplet $(b,0,\nu)$, with $b$ chosen as in (\ref{con30}). 
The call option prices $\bbe\left(e^{X_{t}}-e^{\kappa}\right)_{+}$ can be computed using Fourier inversion formulas or, in some particular cases, via {\Green analytic} closed form formulas.}

\section{Asymptotics of the implied volatility}
Using the leading term of the time-$t$ price for the {\PaleGrey out-of-the-money}
call option as computed in the previous section, we now obtain the
asymptotic behavior of the implied volatility $\hat{\sigma}(t;s)$
near $T$. {\DRed This} is defined implicitly by the equation
\begin{equation}\label{BS1}
C(t,s)=C_{BS}(t,s;\hat{\sigma}(t;s),r),
\end{equation}
where $C_{BS}(t,s;\sigma,r)$ is the classical time-$t$ Black-Scholes
call-option price  corresponding to an interest rate $r$, a
volatility $\sigma$, and time-$t$ spot price $s$. First, we need the following well-known result (see, e.g., Lemma 2.5 in \cite{Gatheral:2009}){\PaleGrey .}
\begin{lem}
Let $C_{BS}(t,s;\sigma,r)$ be the classical Black-Scholes call price
function. Then, as $t\uparrow T$,
\begin{align}\label{OptBS1}
C_{BS}(t,s;\sigma,r)&=\frac{1}{\sqrt{2\pi}}\frac{K\sigma^{3}(T-t)^{3/2}}{(\ln
K-\ln s)^{2}}\exp\left(-\frac{(\ln K-\ln
s)^{2}}{2\sigma^{2}(T-t)}\right)\\
&\qquad{}\times\exp\left(-\frac{\ln K-\ln s}{2}+\frac{r(\ln K-\ln
s)}{\sigma^{2}}\right)+R(t,s;\sigma,r).\nonumber
\end{align}
The remainder satisfies
\begin{equation}
|R(t,s;\sigma,r)|\leq M(T-t)^{5/2}\exp\left(-\frac{(\ln K-\ln
s)^{2}}{2\sigma^{2}(T-t)}\right),
\end{equation}
where $M=M(s,\sigma,r,K)$ is uniformly bounded if all the indicated
parameters vary in a bounded region.
\end{lem}
The next result gives the asymptotic behavior of $\hat{\sigma}(t,s)$. This has already been obtained for a
pure-L\'evy processes (see, e.g., \cite{Tankov} and
\cite{FigForde:2010}) and is presented here for the sake of
completeness{\PaleGrey .}
\begin{prop}\label{ImplVol1stOrd}
{\DRed Let the dynamics of $Z$ be given by (\ref{jumpdif}) and (\ref{jumpdifb}), and the
conditions of both Theorem \ref{jumpdifdist} and Corollary
\ref{IDynkinSV} as well as (\ref{con320}) be satisfied.}
Let $\hat{\sigma}(t;s)$ be the implied volatility when {\Red the time-$t$ stock
price $S_{t}$ is $s$}. Then, as $t\uparrow T$,
\begin{equation}\label{imvol}
\hat{\sigma}^{2}(t;s)\sim\frac{(\ln K-\ln s)^{2}}{-2(T-t)\ln (T-t)}.
\end{equation}
\end{prop}
\begin{proof}
{\DRed Let $\kappa:=\ln K/s$ be the log-moneyness and let $\tau:=T-t$ be the time-to-maturity, where, for simplicity, we write $\hat\sigma(t)$ instead of $\hat\sigma(t;s)$.}
Using the leading terms in (\ref{OptExp1}) and (\ref{OptBS1}), we
obtain that as $t\uparrow T$:
\begin{align}\label{sim1}
\tau u(s,K)\!\sim
v(s,K)\hat{\sigma}^{3}(t)\tau^{3/2}\exp\left(-\frac{\kappa^{2}}{2\hat{\sigma}^{2}(t)\tau}+\frac{r\kappa}{\hat{\sigma}^{2}(t)}\right),
\end{align}
where
\begin{align*}
u(s,K)=\int_{\bbr}\left(se^{u}-K\right)_{+}s(u)du,\qquad 
v(s,K)=\frac{1}{\sqrt{2\pi}}\frac{K}{\kappa^{2}}e^{-\frac{\kappa}{2}}.
\end{align*}
Assume that $\limsup_{t\uparrow
T}\hat{\sigma}(t)\tau^{1/2}=c\in(0,+\infty)$, then $\limsup_{t\uparrow T}\hat{\sigma}(t)=+\infty$ and, thus,
\begin{align}
&\limsup_{t\uparrow
T}\big(\hat{\sigma}(t)\tau^{1/2}\big)^{3}
\exp\left(-\frac{\kappa^{2}}{2\hat{\sigma}^{2}(t)\tau}+\frac{r\kappa}{\hat{\sigma}^{2}(t)}\right)=c^{3}\exp{\left(-\frac{\kappa^{2}}{2c^{2}}\right)}\neq
0\nonumber.
\end{align}
So the right hand side of (\ref{sim1}) does not converge to 0 while
the left hand side does, which is clearly a contradiction.
Now if
\(
	\limsup_{t\uparrow
T}\hat{\sigma}(t)\tau^{1/2}=+\infty,
\)
then $\limsup_{t\uparrow
T}\hat{\sigma}(t)=+\infty$ and, thus,
\begin{eqnarray}
\limsup_{t\uparrow
T}\big(\hat{\sigma}(t)\tau^{1/2}\big)^{3}\exp\left(-\frac{\kappa^{2}}{2\hat{\sigma}^{2}(t)\tau}+\frac{r\kappa}{\hat{\sigma}^{2}(t)}\right)=+\infty\nonumber.
\end{eqnarray}
Again we obtain the same contradiction.
Therefore, we have $\limsup_{t\uparrow
T}\hat{\sigma}(t)\tau^{1/2}=0${\PaleGrey .} Then, (\ref{sim1})
can now be equivalently written as
\begin{align}
\exp\!\left(\!-\frac{\kappa^{2}}{2\hat{\sigma}^{2}(t)\tau}+\frac{r\kappa}{\hat{\sigma}^{2}(t)}+3\ln
\left(\hat{\sigma}(t)\tau^{1/2}\right)-\ln
\tau\right)\!\sim\!\frac{u(s,K)}{v(s,K)}.\nonumber
\end{align}
Hence, as $t\uparrow T$,
\[
	\lim_{t\uparrow{}T}
	\left(\!-\frac{\kappa^{2}}{2\hat{\sigma}^{2}(t)\tau}+\frac{r\kappa}{\hat{\sigma}^{2}(t)}+3\ln
\left(\hat{\sigma}(t)\tau^{1/2}\right)-\ln
\tau -\ln \frac{u(s,K)}{v(s,K)}\right)=0
\]
Finally, since
\(
\lim_{t\uparrow T}\hat{\sigma}^{2}(t)\tau\ln
\left(\hat{\sigma}(t)\tau^{1/2}\right)=0,
\) 
 we obtain (\ref{imvol}).
\end{proof}

\begin{rem}
\hfill
\begin{enumerate}
	\item {\DRed As seen in Proposition \ref{ImplVol1stOrd}, the leading order term of the implied volatility near expiration is ``model free", i.e.,  it does not depend on any of the model parameters.}
	\item {\DRed As discussed in Remark \ref{CnnexLevy}, the second order expansion of OTM call option prices coincides with that of a purely exponential L\'evy model with volatility parameter $\sigma$ equal to the spot volatility $\sigma(Y_{t})$. Thus, the second-order expansion for the implied volatility given in \cite{FigForde:2010} applies:
\begin{equation*}
\hat{\sigma}^2(t;s) = v_0(\tau;\kappa) \left(1+
v_1(\tau;\kappa) \,+\, o\left(\frac{1}{\log\frac{1}{\tau}}\right) \right),
 \quad \quad \quad \quad (\tau \to 0),
\end{equation*}
where $\kappa$ and $\tau$ denote respectively the spot log-moneyness $\kappa:=\log K/s$ and time-to-maturity $\tau:=T-t$, while $v_{0}$ and $v_{1}$ are given by
\begin{align*}
v_0(\tau;\kappa)&= \frac{\frac{1}{2}\kappa^2}{-\tau \log \tau}\,, \nonumber \\
v_1(\tau;\kappa)&=  \frac{1}{\log(\frac{1}{\tau})}\log\left(\frac{4 \sqrt{\pi}e^{-\kappa/2}}{\kappa}\int (e^{u}-e^{\kappa})_{+}s(u)du \log^{3/2}\left(\frac{1}{\tau}\right)\right). \nonumber
\end{align*}
}
	\item {\DRed In the very recent manuscript \cite{GL11}, the authors give a blueprint to generate expansions of arbitrary order for the implied volatility $\hat\sigma^{2}(t;s)$. Interestingly,  such expansions are determined exclusively by the leading order term of the option price expansion, meaning that the stochastic volatility correction is not relevant for implied volatility\footnote{{\DRed We thank an anonymous referee for bringing our attention to this manuscript.}}.}
\end{enumerate}
\end{rem}

\section{Small-time expansions for the L\'evy transition densities}

In this part, we revisit the important problem of finding small-time
expansions for the transition densities of L\'evy processes. This
problem has recently been considered in \cite{Ruschendorf} and also in \cite{FigHou:2008}. As in Section
\ref{Sect:Not}, we consider a general L\'evy process $X$ with L\'evy
triplet $(b,\sigma^{2},\nu)$. It is well-known that under general
conditions (see, e.g, \cite{Leandre} and \cite{Picard:1997}):
\begin{equation}\label{LDF}
    \lim_{\t\to{}0}\frac{1}{\t}\f_{\t}(x)=s(x), \quad (x\neq 0),
\end{equation}
where  $\f_{\t}$ denotes the probability density of $X_{t}$ and $s$
is the L\'evy density of $\nu$ (both densities are assumed to
exist). In many applications, the following uniform convergence
result is more desirable
\begin{equation}\label{LDF2}
    \lim_{\t\to{}0}\sup_{|x|\geq{}\eta} \left|\frac{1}{\t}\f_{\t}(x)-\s(x)\right|=0,
\end{equation}
for a fixed $\eta>0$.
The limit (\ref{LDF2}) is related to the general expansions of the transition densities:
\begin{equation}\label{GED}
 \f_{t}(x)=\sum_{k=1}^{n} \frac{\a_{k}(x)}{k!} t^{k}+ t^{n+1}O_{\eta}(1),
\end{equation}
{\Green which is} valid for any $|x|\geq{}\eta$ and $0<t<t_{0}$, with $t_{0}$ possibly depending on the given $\eta>0$ and $n\geq{}0$. Above, {\DRed and as in (\ref{FIntSpecONt}),} $O_{\eta}(1)$ denotes a function of $x$ and $t$ such that
\[
    \sup_{0<t\leq t_{0}}\sup_{|x|\geq{}\eta}|O_{\eta}(1)|<\infty.
\]
Note that (\ref{LDF2}) follows from (\ref{GED}) when $n=1$ and $\a_{k}(x)=\s(x)$.

{\Red The expansion (\ref{LDF2}) was first proposed in \cite{Ruschendorf}} building on results of
\cite{Leandre}, where {\Red the pointwise convergence in (\ref{LDF}) was obtained}. In both papers, the standing {\Green assumptions} on the
L\'evy density $\s$ of the L\'evy process $X$ are:
 \begin{align}\label{C1E}
    &{\Red {\bf (i)}\;\;\text{There exists } 0<\alpha<2},\text{ such that } \liminf_{\eta\to{}0} \eta^{\alpha-2}\int_{-\eta}^{\eta}z^{2}s(z)dz>0; \\
    \label{C1Ea}
    &{\Red {\bf (ii)}\;\;}s\in C^{\infty}(\bbr\setminus\{0\});\\
    \label{C1Eb}
    &{\Red {\bf (iii)}\;\;}\int_{|z|\geq{}\eta}\frac{|s'(z)|^{2}}{s(z)}dz<\infty, \; {\Red \text{ for all }} \eta>0; \\
    \label{C2E}
    &{\Red {\bf (iv)}\;\;\text{There exists }} h \in C^{\infty}\text{ such that }h(z)=O(z^{2}) \;\;( z\to{}0),\\
    &\qquad h(z)>0
    \text{ if }s(z)>0,\;\text{ and } \int_{|z|\leq{}1}\left|\frac{d}{dz}h(z)\s(z)\right|^{2}\frac{1}{\s(z)} dz<\infty.
    \nonumber
 \end{align}
Condition (\ref{C1E}) is used to conclude the existence of a
$C^{\infty}$ transition density $\f_{t}$ (see \cite[Chapter 5]{Sato:1999}),
while (\ref{C1Ea})-(\ref{C2E}) are needed to establish an estimate
for the transition density using Malliavin Calculus. However, the method of proof of
\cite{Ruschendorf} {\Green is not convincing} {\PaleGrey and} can only {\PaleGrey yield} the first order expansion in (\ref{GED}) (see the introduction of
\cite{FigHou:2008} for more details). {\PaleGrey Recently, \cite{FigHou:2008}} obtained (\ref{GED}) under the following two hypotheses:
\begin{align}
\label{C3E}&\gamma_{\eta,k}:=\sup_{|x|\geq{}\eta} |\s^{(k)}(x)|<\infty,\\
\label{C3F}&\limsup_{t\searrow{}0}\sup_{|x|\geq{}\eta}
|\f_{t}^{(k)}(x)|<\infty,\quad  {\text{ for all } k\geq{}0\text{ and for all } \eta>0}.
\end{align}
Condition (\ref{C3E}) is quite mild but condition (\ref{C3F}) could
be hard to prove in general since closed-form expressions of the densities $\f_{t}$ are lacking. Nevertheless
\cite{FigHou:2008} shows that condition (\ref{C3F}) is satisfied by,
e.g., the CGMY model of \cite{Madan} or Koponen~\cite{Koponen} and
by other types of tempered stable L\'evy processes (as defined in
\cite{Rosinski:2007}).

In this section, we show that (\ref{C3F}) is not {needed} to obtain
(\ref{GED}).
See Appendix \ref{ApC} for the proof of the following result{\PaleGrey .}
\begin{thm}\label{ExpTrDnsty}
Let $\eta>0$ and  $n\geq 1$, and let the conditions
(\ref{C1E})-(\ref{C3E}) be satisfied. Then, (\ref{GED}) holds true
for all $0<t\leq 1$ and $|x|\geq \eta$. Moreover,  there exists
$\varepsilon_{0}(\eta,n)>0$ such that for all
$0<\varepsilon<\varepsilon_{0}$, the coefficients $\a_{k}$ admit the
following representation (which is moreover {independent of $\varepsilon$} for any
$0<\varepsilon<\varepsilon_{0}$):
\begin{equation}\label{Cnst11}
        \a_{k}(x):=
        \sum_{j=1}^{k}  \binom{k}{j}(-\lambda_{\varepsilon})^{k-j} \sum_{i=1}^{j} \binom{j}{i} L_{\varepsilon}^{j-i}\hat{s}_{i,x}(0),
\end{equation}
where $\hat{s}_{i,x}(u):= \bar{s}_{\varepsilon}^{*i}(x-u)$.
\end{thm}

\begin{rem}
Combining the proofs of Theorem \ref{jumpdifdist} and Theorem
\ref{ExpTrDnsty}, it is possible to obtain a small-time expansion
for the jump-diffusion model (\ref{jumpdif})-(\ref{jumpdifb})
assuming, for instance, that the stochastic volatility model admits
a density function $d_{t}$ satisfying the small-time estimate:
    \[
        \sup_{|x|\geq{}\eta}d_{t}(x)\leq{} M_{p,\eta}t^{p},
    \]
    for any $p\geq{}1$ and $0<t<t_{0}(p,\eta)$ and some constant $M_{p,\eta}<\infty$.
\end{rem}

\medskip
\noindent
{\DRed \textbf{Acknowledgments:} 
The authors are grateful to two anonymous referees for their constructive and insightful comments that greatly helped to improve {\PaleGrey the} paper.}

\appendix
\renewcommand{\theequation}{A-\arabic{equation}}
\setcounter{equation}{0}  

\section{Proof of Lemma \ref{prop:IDynkin}} \label{ApA}
Let us show (\ref{IDynkin}) for $n=1$ (the other cases are easily obtained by induction).
First, applying It\^o's lemma (\cite[Theorem
I.4.56]{Shiryaev:2003}),
\begin{align*}
    g(X_{t})
        &=g(0)+\int_{0}^{t}  Lg (X_{u}) du+\sigma \int_{0}^{t} g'(X_{u})d W_{u}\\
        &\quad+\int_{0}^{t}\int_{\bbr} {\Blue \left(g(X_{u^{-}}+z)-g(X_{u^{-}})\right)}\bar\mu(du,dz),
\end{align*}
where {\Blue $L$ is given by (\ref{InfGenLevy0}).} One can easily check that $Lg(x)$ is well-defined {in light of the}
continuity of $g^{(i)}$ and (\ref{NCID}). Indeed, there {exist
constants $M_{i},i=0,1,2$,} such that
$|g^{(i)}(x)|\leq{}M_{i}e^{\frac{c}{2}|x|}$, for all $x$, and thus,
\begin{align*}
    &\left|\int_{|z|\geq{}1} g(x+z)\nu(dz)\right|\leq{} {\Blue M_{0}}\int_{|z|\geq{}1} e^{\frac{c}{2} |z|}\nu(dz) e^{\frac{c}{2}|x|},\\
    &\left|\int_{|z|\leq{}1} (g(x+z)-g(x)-g'(x)z)\nu(dz)\right|\leq{}M_{2} e^{\frac{c}{2}} \int_{|z|\leq{}1} z^{2}\nu(dz) e^{\frac{c}{2}|x|}.
\end{align*}
Next, we show that the last two terms of the expansion of $g(X_{t})$
above are true martingales. Indeed, it suffices that
\begin{align}\label{B1}
    &\bbe \int_{0}^{t} \left|g'(X_{u})\right|^{2} du <\infty,\\
    \label{B20}
    &\bbe \int_{0}^{t}\int_{|z|>1} \left|g(X_{u}+z)-g(X_{u})\right|\nu(dz)du<\infty,\\
    &{\bbe \int_{0}^{t}\int_{|z|\leq{}1} \left|g(X_{u}+z)-g(X_{u})\right|^{2}\nu(dz)du<\infty}.\label{B2}
\end{align}
Using (\ref{NCID}) and the continuity of $g'$, there exists a
constant $M>0$ such that
\begin{align*}
    \bbe \int_{0}^{t} \left|g'(X_{u})\right|^{2} du\leq M\int_{0}^{t}\bbe e^{c|X_{u}|}du
    \leq M\int_{0}^{t}\bbe e^{c X_{u}}du+\int_{0}^{t}\bbe e^{-c X_{u}}du <\infty,
\end{align*}
for any $t\geq{}0$. Similarly, setting $\bar{B}=\{z:|z|>1\}$,
(\ref{B20}) is satisfied since
\begin{align*}
    \bbe \int_{0}^{t}\int_{\bar{B}} \left|g(X_{u}+z)-g(X_{u})\right|\nu(dz) du
    &\leq\bbe \int_{0}^{t}\int_{\bar{B}} \left|
    \int_{0}^{z} g'(X_{u}+w)d w \right| \nu(dz) du\\
    &\leq M\int_{0}^{t} \bbe e^{c|X_{u}|}du\int_{\bar{B}} \int_{0}^{|z|} e^{c w} d w \nu(dz)<\infty.
\end{align*}
Also, setting $B=\{z:|z|\leq 1\}$,
\begin{align*}
    \bbe \int_{0}^{t}\int_{B} \left|g(X_{u}+z)-g(X_{u})\right|^{2}\nu(dz)du&\leq \bbe \int_{0}^{t}\int_{B} \int_{0}^{1} |g'(X_{u}+z\beta)|^{2}d\beta z^{2}\nu(dz) du\\
    &\leq \int_{0}^{t} \bbe e^{c|X_{u}|}du \int_{B} \int_{0}^{1} e^{c|z|\beta}d\beta z^{2}\nu(dz)<\infty.
\end{align*}
We then have that
\begin{align*}
    \bbe g(X_{t})=g(0)+\bbe \int_{0}^{t}  Lg (X_{u}) du,
\end{align*}
which leads to (\ref{Dynkin}), provided \(
    \int_{0}^{t}\bbe \left|Lg(X_{u})\right|du<\infty.
\) The {\Red latter} is proved using (\ref{NCID}) and arguments as above.

\noindent
In order to obtain (\ref{IDynkin}) for $n=1$ by iterating
(\ref{Dynkin}), we need to show that for any $C^{4}$ function $g$
satisfying (\ref{NCID}),
\begin{equation}\label{LNCID}
\limsup_{|y|\to\infty}\;e^{-\frac{c}{2}|y|}|(Lg)^{(i)}(y)|<\infty,
\end{equation}
for $i=0,1,2$. To this end, we first note that
\begin{equation}
(Lg)^{(i)}(y)=bg^{(i+1)}(y)+\frac{\sigma^{2}}{2}g^{(i+2)}(y)+\int_{\bbr}(g^{(i)}(y+z)-g^{(i)}(y)-zg^{(i+1)}(y){\bf
1}_{|z|\leq 1})\nu(dz)\nonumber
\end{equation}
for $i=0, 1, 2$. Hence, it is sufficient to show (\ref{LNCID}) when
$i=0$. {\Red But,} 
\begin{align}
\label{lgeq1}
e^{-\frac{c}{2}|y|}|Lg(y)| &\leq  b
e^{-\frac{c}{2}|y|}|g'(y)|+\frac{\sigma^2}{2}e^{-\frac{c}{2}|y|}|g''(y)|\\
\label{lgeq2}&\quad+e^{-\frac{c}{2}|y|}\int_{|z|>1}|g(y+z)-g(y)|\nu(dz)\\
\label{lgeq3}&\quad+e^{-\frac{c}{2}|y|}\int_{|z|\leq
1}|g(y+z)-g(y)-zg'(y)|\nu(dz).
\end{align}
The {\PaleGrey limits} of each term of the right-hand terms in (\ref{lgeq1}) are trivially finite as $|y|\rightarrow\infty$ by the assumption
(\ref{NCID}). For the term in (\ref{lgeq2}), again by the assumption
(\ref{NCID}) and the continuity of $g^{(i)}$, there exists $M>0$
such that,
\begin{equation}
|g^{(i)}(y)|\leq M e^{\frac{c}{2}|y|},\quad i=0,1,2.\nonumber
\end{equation}
It then follows that
\begin{align*}
e^{-\frac{c}{2}|y|}\int_{|z|>1}\!\!|g(y+z)-g(y)|\nu(dz)&=e^{-\frac{c}{2}|y|}\int_{|z|>1}\Big|\int_{0}^{z}g'(y+w)\big|dw\Big|\nu(dz)\nonumber\\
&\leq M \int_{|z|>1}\!\!\Big(\!\int_{0}^{|z|}e^{\frac{c}{2}w}dw\!\Big)\nu(dz)\nonumber\\
&=M\int_{|z|>1}e^{\frac{c}{2}|z|}\nu(dz)<\infty,\nonumber
\end{align*}
which immediately implies that
\begin{equation}
\limsup_{|y|\rightarrow\infty}\;e^{-\frac{c}{2}|y|}\!\!\int_{|z|>1}|g(y+z)-g(y)|\nu(dz)<\infty.
\end{equation}
Similarly, we can show that the limit as
$|y|\rightarrow\infty$ of (\ref{lgeq3}) is finite. Therefore, we can iterate
(\ref{Dynkin}) to obtain (\ref{IDynkin}) for $n=1$. \hfill $\Box$

\section{Proof of Theorem \ref{jumpdifdist}} \label{ApB}
\renewcommand{\theequation}{B-\arabic{equation}}

We analyze each term on the right-hand side of the expansion of
$\bbe f(Z_{t})$ given in {(\ref{T1b})-(\ref{T2b})}:

\medskip
\noindent
\textbf{(1)} For any $z\geq z_{0}$, we have
\begin{equation}\label{TailUX}
\bbe f_{z}\left(U_{t}+X^{\varepsilon}_{t}\right)=\bbp(U_{t}+X_{t}^{\varepsilon}\geq{}z)\leq
    \bbp(U_{t}\geq{}z/2)+\bbp(X_{t}^{\varepsilon}\geq{}z/2).
\end{equation}
By our assumption (\ref{tailU0}), there exists $t_{0}(z_{0})>0$ such that for
any $0<t\leq t_{0}$, $z\geq z_{0}>0$,
\begin{equation}\label{TailU}
\bbp(U_{t}\geq{}z/2)\leq\bbp(U_{t}\geq{}z_{0}/2)\leq\exp\left(-\frac{d(z_{0}/2)^{2}}{4t}\right),
\end{equation}
which can be seen to be $O_{z_{0}}(t^{n+1})$. Also, the second term on the right-hand-side of
(\ref{TailUX}) is $O_{\varepsilon,z_{0}}(t^{n+1})$ in light of
(\ref{TailEstLevy}) by taking $a:=(n+1)/z_{0}$ and since $0<\varepsilon<z_{0}/(n+1)\wedge 1$.

\medskip
\noindent
\textbf{(2)} The second term in (\ref{T1b}) is also
$O_{\varepsilon,z_{0}}(t^{n+1})$ since $f\leq 1$ and clearly
\[
	{\PaleGrey e^{-\lambda_{\varepsilon} t}\sum_{k=n+1}^{\infty}(\lambda_{\varepsilon}t)^{k}/k!\leq (\lambda_{\varepsilon} t)^{n+1}=O(t^{n+1})}.
\]

\medskip
\noindent \textbf{(3)} We proceed to deal with the terms in
(\ref{T2b}). By the independence of $U$ and $X$,
\begin{equation}
    \bbe f_{z}\left(U_{t}+X^{\varepsilon}_{t}+\sum_{i=1}^{k}\xi_{i}\right)=\bbe
    \widetilde{f}_{k,z}\left(U_{t}+X^{\varepsilon}_{t}\right)
    =\bbe\breve{f}_{k,z,t}(X_{t}^{\varepsilon}),
\end{equation}
where
\[
   \widetilde{f}_{k,z}(y)\!:=\!(\lambda_{\varepsilon})^{-k}\!\!\int_{z-y}^{\infty}\bar{s}_{\varepsilon}^{*k}(u)du\quad\text{ and }\quad\breve{f}_{k,z,t}(y)\!:=\!\bbe\widetilde{f}_{k,z}\left(U_{t}+y\right).
\]
In particular, by the assumption (\ref{conB2}),
\begin{align*}
    &\widetilde{f}_{k,z}^{(j)}(y)=
    (\lambda_{\varepsilon})^{-k}(-1)^{j-1} \bar{s}_{\varepsilon}^{*(k-1)}*\bar{s}_{\varepsilon}^{(j-1)}(z-y),\\
    &\sup_{y,z}\left|\widetilde{f}_{k,z}^{(j)}(y)\right|\leq
\lambda_{\varepsilon}^{-1}\|
\bar{s}_{\varepsilon}^{(j-1)}\|_{\infty}\leq \lambda_{\varepsilon}^{-1} \max_{0\leq i\leq{}j-1}\gamma_{i,\varepsilon/2}:=\Gamma_{\varepsilon}<\infty.
\end{align*}
It follows that $\breve{f}_{k,z,t}\in C^{\infty}_{b}(\bbr)$ and
moreover,
\begin{align}
    \breve{f}_{k,z,t}^{(j)}(y)=\bbe\widetilde{f}_{k,z}^{(j)}\left(U_{t}+y\right),
    \text{ and }\sup_{z,y}\left|\breve{f}_{k,z,t}^{(j)}(y)\right|\leq \Gamma_{\varepsilon},\text{ for any }j\geq 0.\label{bdbrevef}
\end{align}
We can thus apply the iterated formula (\ref{IDynkin})
to get
\begin{equation}\label{expbrevef}
\bbe\breve{f}_{k,z,t}(X_{t}^{\varepsilon})=\sum_{i=0}^{n-k}\frac{t^{i}}{i!}L_{\varepsilon}^{i}\breve{f}_{k,z,t}(0)+\frac{t^{n-k+1}}{(n-k)!}\int_{0}^{1}(1-\alpha)^{n-k}\bbe\{L_{\varepsilon}^{n-k+1}\breve{f}_{k,z,t}(X_{\alpha
t}^{\varepsilon})\}d\alpha.
\end{equation}
It follows from the representation in Lemma \ref{RIGen} and
(\ref{bdbrevef}) that
\begin{equation}
    \sup_{z}
     \int_{0}^{1}(1-\alpha)^{n-k}\bbe\big(L_{\varepsilon}^{n-k+1}\breve{f}_{k,z,t}(X_{\alpha
t}^{\varepsilon})\big)d\alpha <\infty,\nonumber
\end{equation}
and thus the second term on the right hand side of (\ref{expbrevef})
is $O_{\varepsilon, z_{0}}(t^{n-k+1})$.

\smallskip
\noindent
\textbf{(4)} Combining (\ref{T1b}), (\ref{T2b})
and (\ref{expbrevef}), we obtain
\begin{align*}
\bbe
f(Z_{t})&=e^{-\lambda_{\varepsilon}t}\sum_{k=1}^{n}\frac{(\lambda_{\varepsilon}t)^{k}}{k!}\bbe\breve{f}_{k,z,t}(X_{t}^{\varepsilon})+O_{\varepsilon,
z_{0}}(t^{n+1})\\
&=e^{-\lambda_{\varepsilon}t}\sum_{k=1}^{n}\frac{(\lambda_{\varepsilon}t)^{k}}{k!}\sum_{i=0}^{n-k}\frac{t^{i}}{i!}L_{\varepsilon}^{i}\breve{f}_{k,z,t}(0)+O_{\varepsilon,
z_{0}}(t^{n+1})\\
&=e^{-\lambda_{\varepsilon}t}\sum_{j=1}^{n}\frac{t^{j}}{j!}\sum_{k=1}^{j}
\binom{j}{k}\lambda_{\varepsilon}^{k}L_{\varepsilon}^{j-k}\breve{f}_{k,z,t}(0)+O_{\varepsilon,z_{0}}(t^{n+1}).
\end{align*}
Using again the representation in Lemma \ref{RIGen} and
(\ref{bdbrevef}), it follows that
\[
    L_{\varepsilon}^{j-k}\breve{f}_{k,z,t}(x)=L_{\varepsilon}^{k}\left[\bbe\widetilde{f}_{k,z}(U_{t}+\cdot)\right](x)=\lambda_{\varepsilon}^{-k}
     L_{\varepsilon}^{k}\left[\bbe\widehat{f}_{k,z}(U_{t}+\cdot)\right](x),
\]
and (\ref{jumdifexpan1}) is obtained.\hfill
$\Box$

\section{Proof of Theorem \ref{jumpdifopt}} \label{ApC}
\renewcommand{\theequation}{C-\arabic{equation}}
\noindent We analyze each term in (\ref{T1b}) and (\ref{T2b}) {\Blue through the following steps:}

\smallskip
\noindent
\textbf{(1)} For $z\leq z_{0}<0$,
\begin{align}
\label{expfUX}\bbe f_{z}(U_{t}+X^{\varepsilon}_{t})&=\bbe
\left(e^{z+U_{t}+X^{\varepsilon}_{t}}-1\right)_{+}\leq\bbe\left(e^{U_{t}+X^{\varepsilon}_{t}}{\bf
1}_{\{U_{t}+X^{\varepsilon}_{t}\geq -z\}}\right)\\
&\leq \left(\bbe e^{2U_{t}+2X^{\varepsilon}_{t}}\bbp
(U_{t}+X^{\varepsilon}_{t}\geq -z)\right)^{1/2}\nonumber\\
&\leq \left(\bbe e^{2U_{t}}\bbe
e^{2X^{\varepsilon}_{t}}\right)^{1/2}\Big(\bbp (U_{t}\geq
-z/2)+\bbp(X^{\varepsilon}_{t}\geq
-z/2)\Big)^{1/2},\nonumber\\
&={e^{t\psi(2)/2}}\left(\bbe e^{2U_{t}}\right)^{1/2}\Big(\bbp
(U_{t}\geq -z/2)+\bbp(X^{\varepsilon}_{t}\geq
-z/2)\Big)^{1/2},\nonumber
\end{align}
where {$\psi$} is the characteristic exponent of $X^{\varepsilon}$.
Since $M_{t}:=e^{U_{t}}$ satisfies the SDE $d M_{t} = M_{t}
\sigma(Y_{t})dW^{(1)}_{t}$, and {from the Burkh\"older-Davis-Gundy
inequality},
\begin{align}
\bbe e^{2U_{t}}&=\bbe \left(1+\int_{0}^{t} M_{s} \sigma(Y_{s}) dW^{(1)}_{s}\right)^{2}\nonumber\\
&\leq
2+2\bbe\left(\int_{0}^{t}e^{U_{s}}\sigma(Y_{s})dW^{(1)}_{s}\right)^{2}\leq
2+2M^{2}\,\bbe\int_{0}^{t}e^{2U_{s}}ds. \nonumber
\end{align}
By Gronwall's {\PaleGrey inequality},
\begin{align}
\bbe e^{2U_{t}}\leq 2e^{2M^{2}t}=O_{\varepsilon,z_{0}}(1)\nonumber.
\end{align}
Therefore, the right-hand-side of (\ref{expfUX}) {is of order}
$o_{\varepsilon,z_{0}}(t^{n+1})$ by (\ref{TailEstLevy}) and
(\ref{tailU0}).

\smallskip
\noindent\textbf{(2)} The second summation in (\ref{T1b}) is also
$O_{\varepsilon,z_{0}}(t^{n+1})$ since for any $k\geq n+1$,
\begin{align}
\bbe f_{z}(U_{t}+X^{\varepsilon}_{t}+\sum_{i=1}^{k}\xi_{i})&\leq
e^{z}\bbe e^{U_{t}}\bbe e^{X^{\varepsilon}_{t}}(\bbe
e^{\xi_{1}})^{k}\leq
\lambda_{\varepsilon}^{-k}e^{t\Psi(1)}\left(\int_{\bbr}e^{x}\bar{s}_{\varepsilon}(x)dx\right)^{k}.\nonumber
\end{align}

\noindent
\textbf{(3)} To {study} the summation in (\ref{T2b}), recall
that by the independence of $U$ and $X$, for any $1\leq k\leq n$,
\begin{align}
\bbe
f_{z}\left(U_{t}+X^{\varepsilon}_{t}+\sum_{i=1}^{k}\xi_{i}\right)=\bbe
\widetilde{f}_{k,z}\left(U_{t}+X^{\varepsilon}_{t}\right)=\bbe
\breve{f}_{k,z,t}(X_{t}^{\varepsilon}),\nonumber
\end{align}
where
\begin{align}
\breve{f}_{k,z,t}(x)=\bbe
\widetilde{f}_{k,z}\left(U_{t}+x\right)\quad\text{ and
}\quad\widetilde{f}_{k,z}(x)=\bbe
f_{z}\left(x+\sum_{i=1}^{k}\xi_{i}\right). \nonumber
\end{align}
Let us show that $\tilde{f}_{k,z}$ is $C^{\infty}$. Indeed, since
\begin{align*}
    \widetilde{f}_{k,z}(x)=\lambda_{\varepsilon}^{-k}
    \int_{\bbr^{k-1}} \int_{-\sum_{\ell=2}^{k}u_{\ell}-z-x}^{\infty}
    \left(e^{z+x+\sum_{\ell =1}^{k}u_{\ell}}-1\right)
    \bar{s}_{\varepsilon}(u_{1}) du_{1}\prod_{\ell=2}^{k} \bar{s}_{\varepsilon}(u_{\ell}) du_{\ell},
\end{align*}
and $\bar{s}_{\varepsilon}\in C^{\infty}_{b}$,
we have that
\begin{align*}
    \widetilde{f}'_{k,z}(x)&=\lambda_{\varepsilon}^{-k}\int_{\bbr^{k-1}} \int_{-\sum_{\ell=2}^{k}u_{\ell}-z-x}^{\infty}e^{z+x+\sum_{\ell =1}^{k}u_{\ell}}
    \bar{s}_{\varepsilon}(u_{1}) du_{1}\prod_{\ell=2}^{k} \bar{s}_{\varepsilon}(u_{\ell}) du_{\ell},\\
    \widetilde{f}''_{k,z}(x)&=\lambda_{\varepsilon}^{-k}\int_{\bbr^{k-1}} \int_{-\sum_{\ell=2}^{k}u_{\ell}-z-x}^{\infty}e^{z+x+\sum_{\ell =1}^{k}u_{\ell}}
    \bar{s}_{\varepsilon}(u_{1}) du_{1}\prod_{\ell=2}^{k} \bar{s}_{\varepsilon}(u_{\ell}) du_{\ell}\\
    &+\lambda_{\varepsilon}^{-k}\int_{\bbr^{k-1}}
    \bar{s}_{\varepsilon}\left(-\sum_{\ell =2}^{k}u_{\ell}-z-x\right) \prod_{\ell=2}^{k} \bar{s}_{\varepsilon}(u_{\ell}) du_{\ell}.
\end{align*}
Using induction, we see that
 \begin{align}\label{tildefki}
    \widetilde{f}^{(i)}_{k,z}(x)&=\lambda_{\varepsilon}^{-k}
    \int_{\bbr}\int_{-\sum_{\ell=2}^{k}u_{\ell}-z-x}^{\infty}e^{z+x+\sum_{\ell =1}^{k}u_{\ell}}
    \bar{s}_{\varepsilon}(u_{1}) du_{1}\prod_{\ell=2}^{k} \bar{s}_{\varepsilon}(u_{\ell}) du_{\ell}\\
    &+\lambda_{\varepsilon}^{-k}\sum_{j=0}^{i-2}(-1)^{j}\int_{\bbr^{k-1}}
    \bar{s}_{\varepsilon}^{(j)}\left(-\sum_{\ell =2}^{k}u_{\ell}-z-x\right)\prod_{\ell=2}^{k} \bar{s}_{\varepsilon}(u_{\ell})
    du_{\ell}.\nonumber
\end{align}
In view of (\ref{conB2}), there exists a constant
$M_{i,\varepsilon}<\infty$ such that, for any $i\geq 1$,
\begin{align}\label{tildefUi}
\left|\widetilde{f}^{(i)}_{k,z}(U_{t}+x)\right|&\leq
\lambda_{\varepsilon}^{-k}\int_{\bbr^{k}}
e^{z+x+\sum_{\ell=1}^{k}u_{\ell}}\prod_{\ell=1}^{k}
\bar{s}_{\varepsilon}(u_{\ell})
du_{\ell}\cdot e^{U_{t}}\\
&\quad+M_{i,\varepsilon}\lambda_{\varepsilon}^{-k}\sum_{j=0}^{i-2}\int_{\bbr^{k-1}}\prod_{\ell=2}^{k}
\bar{s}_{\varepsilon}(u_{\ell})du_{\ell}\max_{0\leq
j\leq{}i}\gamma_{j,\varepsilon/2}.\nonumber
\end{align}
The right-hand side of (\ref{tildefUi}) is integrable since $\bbe
e^{U_{t}}=1$. By dominated convergence, we conclude that
$\breve{f}_{k,z,t}\in C^{\infty}(\bbr)$, and also,
\begin{align*}
\breve{f}^{(i)}_{k,z,t}(x)=\bbe
\left[\widetilde{f}^{(i)}_{k,z}(U_{t}+x)\right],\text{ for all }i\geq
0,\quad\text{and}\quad
\limsup_{|x|\to\infty}e^{-\frac{c}{2}|x|}\left|\breve{f}^{(i)}_{k,z,t}(x)\right|<\infty,
\end{align*}
since $c\geq{}2$. Thus, applying (\ref{IDynkin}) gives
\begin{align}\label{expbrevef2}
\bbe \breve{f}_{k,z,t}(X_{t}^{\varepsilon})=\sum_{i=0}^{n-k}\frac{t^{i}}{i!}L_{\varepsilon}^{i}\breve{f}_{k,z,t}(0)+\frac{t^{n-k+1}}{(n-k)!}\int_{0}^{1}(1-\alpha)^{n-k}\bbe\{L_{\varepsilon}^{n-k+1}\breve{f}_{k,z,t}(X_{\alpha
t}^{\varepsilon})\}d\alpha.
\end{align}
To show that the last integral in (\ref{expbrevef2}) is bounded, we apply Lemma \ref{RIGen} to get that
\begin{align}
\bbe\left((L_{\varepsilon}^{n-k+1}\breve{f}_{k,z,t})(X_{\alpha
t}^{\varepsilon})\right)=
    \sum_{{\bf k}\in\mathcal{K}_{n-k+1}}
        \prod_{i=0}^{4} b_{i}^{k_{i}}\binom{n-k+1}{\bf k}\bbe\left(B_{_{{\bf k},\varepsilon}}\breve{f}_{k,z,t}(X_{\alpha
        t}^{\varepsilon})\right).\nonumber
\end{align}
Thus, it is sufficient to show the boundedness of $\bbe B_{_{{\bf
k},\varepsilon}}\breve{f}_{k,z,t}(X_{\alpha t}^{\varepsilon})$, for
any $1\leq k\leq n$ and ${\bf
k}=(k_{0},\dots,k_{4})\in\mathcal{K}_{n-k+1}$. Indeed, noting that
(\ref{conB1}) implies that
\[
    {\tilde{M}}:=\int_{[0,1]^{k_{3}}\times\bbr^{k_{3}+k_{4}}} e^{\sum_{j=1}^{k_{3}}\beta_{j}w_{j}+\sum_{i=1}^{k_{4}} u_{i}}d\pi_{_{{\bf k},\varepsilon}}<\infty,
\]
we have, for any $x\in\bbr$ and some {\PaleGrey constants} $K_{1},K_{2}<\infty$,
\begin{align}
 \left|B_{_{{\bf k},\varepsilon}}\breve{f}_{k,z,t}(x)\right|& \leq \int_{[0,1]^{k_{3}}\times\bbr^{k_{3}+k_{4}}}
 \left|\breve{f}_{k,z,t}^{(\ell_{\bf k})}\right|\left(x+\displaystyle{\sum_{j=1}^{k_{3}}\beta_{j}w_{j}+\sum_{i=1}^{k_{4}} u_{i}}\right)d\pi_{_{{\bf k},\varepsilon}}\nonumber\\
&\leq \int_{[0,1]^{k_{3}}\times\bbr^{k_{3}+k_{4}}}\bbe\,\left|{\widetilde{f}_{k,z}^{(\ell_{\bf k})}}\right|\left(U_{t}+x+\displaystyle{\sum_{j=1}^{k_{3}}\beta_{j}w_{j}+\sum_{i=1}^{k_{4}} u_{i}}\right)
d\pi_{_{{\bf k},\varepsilon}}\nonumber\\
&\leq{\tilde{M}} \lambda_{\varepsilon}^{-k}\bbe\,
e^{U_{t}}\int_{\bbr^{k-1}} \int_{\bbr}
e^{z+x+\sum_{\ell=1}^{k}u_{\ell}}\bar{s}_{\varepsilon}(u_{1})
du_{1}\prod_{\ell=2}^{k} \bar{s}_{\varepsilon}(u_{\ell})
du_{\ell}\nonumber\\
&\quad+{M_{i,\varepsilon}}
\lambda_{\varepsilon}^{-k}\sum_{j=0}^{\ell_{\bf
k}-2}\int_{\bbr^{k-1}}\prod_{\ell=2}^{k}
\bar{s}_{\varepsilon}(u_{\ell})du_{\ell} \max_{0\leq j\leq{}i}\gamma_{j,\varepsilon/2}\nonumber\\
&={M_{1}}e^{x}+{M_{2}}<\infty,\nonumber
\end{align}
where the third inequality follows from (\ref{tildefUi}). It follows
that $\bbe B_{_{{\bf k},\varepsilon}}\breve{f}_{k,z,t}(X_{\alpha
t}^{\varepsilon})$ is $O_{\varepsilon,z_{0}}(1)$, and so is $\bbe
L_{\varepsilon}^{n-k+1}\breve{f}_{k,z,t}(X_{\alpha
t}^{\varepsilon})$. Therefore, the last integral in
(\ref{expbrevef2}) is indeed $O_{\varepsilon,z_{0}}(t^{n-k+1})$.

\smallskip
\noindent\textbf{(4)} {Plugging} (\ref{expbrevef2}) into (\ref{T1b})
and (\ref{T2b}), and rearranging terms lead to
\begin{align}
\bbe f_{z}(Z_{t})&=e^{-\lambda_{\varepsilon}t}\sum_{k=1}^{n}\frac{(\lambda_{\varepsilon}t)^{k}}{k!}\breve{f}_{k,z,t}(X_{t}^{\varepsilon})+O_{\varepsilon,z_{0}}(t^{n+1})\nonumber\\
\label{expfzZ}&=e^{-\lambda_{\varepsilon}t}\sum_{j=1}^{n}\frac{t^{j}}{j!}\sum_{k=1}^{j}
\binom{j}{k}\lambda_{\varepsilon}^{k}L_{\varepsilon}^{j-k}\breve{f}_{k,z,t}(0)+O_{\varepsilon,z_{0}}(t^{n+1}).
\end{align}
It remains to expand the coefficients
\begin{align}\label{relatbrevefhatf}
L_{\varepsilon}^{j-k}\breve{f}_{k,z,t}(0)=L_{\varepsilon}^{j-k}\left[\bbe\widetilde{f}_{k,z}(U_{t}+\cdot)\right](0)=\lambda_{\varepsilon}^{-k}L_{\varepsilon}^{j-k}\left[\bbe\widehat{f}_{k,z}(U_{t}+\cdot)\right](0).
\end{align}
Using the expansion (\ref{IDynkinDiff}) and Remark \ref{IDynkinSV2}, we have
\begin{align}
\bbe\widehat{f}_{k,z}(U_{t}+x)&=\sum_{i=0}^{n-j}\frac{t^{i}}{i!}\calL^{i}\widehat{f}_{k,z}(x)\!+\!\frac{t^{n-j+1}}{(n\!-\!j\!+\!1)!}\!\int_{0}^{1}\!\!(1\!-\!\alpha)^{n\!-\!j}\bbe\!\left(\!\calL^{n-j+1}\!\widehat{f}_{k,z}\!(U_{\alpha t}+x)\!\right)\!d\alpha\nonumber\\
\label{expwidehatfopt}&=\sum_{i=0}^{n-j}\frac{t^{i}}{i!}\sum_{l=0}^{i}B_{l}^{i}(y_{0}){\Red \calL}_{1}^{l}\widehat{f}_{k,z}(x)\\
&\quad
+\frac{t^{n-j+1}}{(n-j+1)!}\int_{0}^{1}(1-\alpha)^{n-j}\bbe\left(\calL^{n-j+1}\widehat{f}_{k,z}(U_{\alpha
t}+x)\right)d\alpha.\nonumber
\end{align}
Finally, by combining (\ref{expfzZ}), (\ref{relatbrevefhatf}) and
(\ref{expwidehatfopt}), it follows that
\begin{align}
\bbe f_{z}(Z_{t})&=
e^{-\lambda_{\varepsilon}t}\sum_{j=1}^{n}\frac{t^{j}}{j!}\sum_{k=1}^{j}
\binom{j}{k}\Bigg[\sum_{i=0}^{n-j}\frac{t^{i}}{i!}L_{\varepsilon}^{j-k}\left(\sum_{l=0}^{i}B_{l}^{i}(y_{0}){\Red \calL}_{1}^{l}\widehat{f}_{k,z}\right)(0)\nonumber\\
\label{expfzZfinal1}&+\!\frac{t^{n-j+1}}{(n\!-\!j\!+\!1)!}\!\int_{0}^{1}\!\!(1\!-\!\alpha)^{n-j}\bbe\!\left\{\!L_{\varepsilon}^{j-k}\!\big[\!\calL^{n\!-\!j\!+\!1}\!\widehat{f}_{k,z}\!(U_{\alpha
t}\!+\cdot)\big](0)\!\right\}\!d\alpha\!\Bigg]\!+\!O_{\varepsilon,z_{0}}\!(t^{n+1}).
\end{align}
{\Red Finally, since the integral in (\ref{expfzZfinal1}) is $O_{\varepsilon,z_{0}}(1)$, as seen from the uniform boundedness condition (\ref{con320}) and the estimate
(\ref{tildefUi}), we obtain that} 
\begin{align*}
\bbe f_{z}(Z_{t})&=e^{-\lambda_{\varepsilon}t}\sum_{j=1}^{n}\frac{t^{j}}{j!}\sum_{k=1}^{j}
\binom{j}{k}\sum_{i=0}^{n-j}\frac{t^{i}}{i!}L_{\varepsilon}^{j-k}\left(\sum_{l=0}^{i}B_{l}^{i}(y_{0}){\Red \calL}_{1}^{l}\widehat{f}_{k,z}\right)(0)+O_{\varepsilon,z_{0}}(t^{n+1})\\
&=e^{-\lambda_{\varepsilon}t}\!\sum_{j=1}^{n}\frac{t^{j}}{j!}\!\!\sum_{p+q+r=j}\!\binom{j}{p,q,r}\!L_{\varepsilon}^{q}\left(\sum_{m=0}^{r}B_{m}^{r}(y_{0}){\Red \calL}_{1}^{m}\widehat{f}_{p,z}\!\right)(0)\!+\!O_{\varepsilon,z_{0}}(t^{n+1}).
\end{align*}
\hfill $\Box$

\section{Proof of Theorem \ref{ExpTrDnsty}} \label{ApD}
\renewcommand{\theequation}{D-\arabic{equation}}
 We only consider {\Green the case} $x>0$ (the case $x<0$ can be done similarly by considering $\bbp(X_{t}\leq x)$). Again, we start with the expression
\begin{align}\label{T1}
    \bbp(X_{t}\geq{}x)&
    = \underbrace{e^{-\lambda_{\varepsilon}t}\bbp \left(X^{\varepsilon}_{t}\geq{}x \right)}_{B_{t}(x)}
   +\underbrace{
    e^{-\lambda_{\varepsilon}t}\sum_{k=n+1}^{\infty}\frac{(\lambda_{\varepsilon}t)^{k}}{k!}
    \bbp\left(X^{\varepsilon}_{t}+\sum_{i=1}^{k}\xi_{i}\geq{}x\right)}_{C_{t}(x)}\\
    \label{T2}
    &\quad+
    \underbrace{e^{-\lambda_{\varepsilon}t}\sum_{k=1}^{n}\frac{(\lambda_{\varepsilon}t)^{k}}{k!}\bbp \left(X^{\varepsilon}_{t}+\sum_{i=1}^{k}\xi_{i}\geq{}x\right)}_{D_{t}(x)}.
\end{align}
Let $\f_{t}^{\varepsilon}$ denote the density of
$X^{\varepsilon}_{t}$, whose existence follows from (\ref{C1E}).
Given that
\begin{align*}
    &\frac{d}{dx}\bbp \left(X^{\varepsilon}_{t}+\sum_{i=1}^{k}\xi_{i}\geq{}x\right)= - \frac{1}{\lambda_{\varepsilon}^{k}} f_{t}^{\varepsilon} * \bar{s}^{*k}_{\varepsilon}(x),\end{align*}
and that 
\[
	\sup_{x}|f_{t}^{\varepsilon} *  \bar{s}_{\varepsilon}^{*k}(x)|\leq \sup_{x}| \bar{s}_{\varepsilon}^{*k}(x)|\leq \gamma_{\varepsilon/2,0}\lambda_{\varepsilon}^{k-1},
\]
one can interchange derivative and summation in (\ref{T1}) to show that for each $t\geq 0$, $C_{t}$ admits a density $c_{t}$, with moreover,
\begin{equation}\label{Est2}
    \sup_{x}\left|c_{t}(x)\right|=
    \sup_{x} e^{-\lambda_{\varepsilon}t}\sum_{k=n+1}^{\infty}\frac{t^{k}}{k!}f_{t}^{\varepsilon} * \bar{s}^{*k}_{\varepsilon}(x)\leq  e^{-\lambda_{\varepsilon}t}\frac{\gamma_{\varepsilon/2,0}}{\lambda_{\varepsilon}}\sum_{k=n+1}^{\infty}\frac{(\lambda_{\varepsilon}t)^{k}}{k!}\leq\lambda_{\varepsilon}^{n} \gamma_{\varepsilon/2,0} t^{n+1}.
\end{equation}
Also, in view of Proposition III.2 in \cite{Leandre}, there exists a real $\varepsilon_{0}(\eta,n)>0$ such that for all $0<\varepsilon<\varepsilon_{0}$ and $t\leq{}1$,
\begin{equation}\label{LEst}
    \sup_{|x|\geq{}\eta}f_{t}^{\varepsilon}(x)\leq M(\eta,\varepsilon)
    t^{n+1},
\end{equation}
where $M(\eta,\varepsilon)$ is some constant depending only
on $\eta$ and $\varepsilon$. The last step of the proof is to deal with the terms
in $D_{t}$. Recall that
\begin{align*}
    &\bbp\left(X^{\varepsilon}_{t}+\sum_{i=1}^{k}\xi_{i}\geq{}x\right)=\bbe \widetilde{f}_{k,x}
    \left(X^{\varepsilon}_{t}\right),
\end{align*}
and 
\begin{align*}
    &\frac{d^{(i)}}{d z^{i}}\widetilde{f}_{k,x}(y)=\lambda_{\varepsilon}^{-k}
    (-1)^{i-1}\bar{s}_{\varepsilon}^{*(k-1)}*\bar{s}_{\varepsilon}^{(i-1)}(x-y),
\end{align*}
with
\[
     \widetilde{f}_{k,x}(y):=\bbp\left(y+\sum_{\ell=1}^{k}\xi_{i}\geq{}x\right)
     =\lambda_{\varepsilon}^{-k}\int_{x-y}^{\infty}
     \bar{s}_{\varepsilon}^{*k}(u)du.
\]
Then, applying the iterated formula (\ref{IDynkin}), we get
\begin{equation}\label{AE1}
    \bbe \widetilde{f}_{k,x}(X^{\varepsilon}_{t})=
    \sum_{i=0}^{n-k} \frac{t^{i}}{i!}
    L^{i}_{\varepsilon}\widetilde{f}_{k,x}(0)+
    \frac{t^{n+1-k}}{(n-k)!}
    \int_{0}^{1} (1-\alpha)^{n-k}\bbe\left(
    L_{\varepsilon}^{n+1-k} \widetilde{f}_{k,x}(X^{\varepsilon}_{\alpha
    t})\right)d\alpha.
\end{equation}
Using the representation of $L_{\varepsilon}$ in Lemma \ref{RIGen},
one can easily verify that
\begin{align}\label{PDv}
    &\frac{d}{dx}L^{i}_{\varepsilon}\widetilde{f}_{k,x}(y)=- L^{i}_{\varepsilon}\widetilde{f}_{k,x}'(y)=-(\lambda_{\varepsilon})^{-k} L^{i}_{\varepsilon}\hat{s}_{k,x}(y),\\
    &\sup_{x,z}\left|\frac{d}{dx}L_{\varepsilon}^{n+1-k} \widetilde{f}_{k,x}(y)\right|\leq M_{n,k,\varepsilon}
    {\max_{0\leq{}k\leq{}2n}\{\gamma_{\varepsilon/2,k}\}},\label{PDvb}
\end{align}
for some constants $M_{n,k,\varepsilon}<\infty$. Hence, one can pass
$d/dx$ through the integral and the expectation in the last term of
(\ref{AE1}) to get
\begin{equation}\label{AE1b}
    \frac{d}{dx}\bbe \widetilde{f}_{k,x}(X^{\varepsilon}_{t})=-(\lambda_{\varepsilon})^{-k}
    \sum_{i=0}^{n-k} \frac{t^{i}}{i!}
    L^{i}_{\varepsilon}\hat{s}_{k,x}(0)+t^{n+1-k}O_{\varepsilon,k,n}(1),
\end{equation}
where $O_{\varepsilon,k,n}(1)$ {\Green indicates} that
$\sup_{x}|O_{\varepsilon,k,n}(1)|$ is bounded by a constant
depending only on $\varepsilon$, $k$, and $n$. Differentiating
$\bbp(X_{t}\geq{}x)$ in (\ref{T1}) and plugging in (\ref{Est2}),
(\ref{LEst}), (\ref{AE1b}) gives for any
$0<\varepsilon<\varepsilon_{0}$ and $t\leq{}1$,
\[
    \f_{t}(x)=e^{-\lambda_{\varepsilon}t}\sum_{k=1}^{n}\frac{t^{k}}{k!} \sum_{i=0}^{n-k} \frac{t^{i}}{i!}
    L^{i}_{\varepsilon}\hat{s}_{k,x}(0)+t^{n+1}O_{\varepsilon,\eta}(1),
\]
where $O_{\varepsilon,\eta}(1)$ is such that $\sup_{t\leq{}1}\sup_{|x|\geq{}\eta}|O_{\varepsilon,\eta}(1)|<\infty$.
Rearranging the terms above leads to
\[
     \f_{t}(x)=e^{-\lambda_{\varepsilon}t}\sum_{p=1}^{n}c_{p}(x)\frac{t^{p}}{p!}+t^{n+1}O_{\varepsilon,\eta}(1),
\]
with
\[
    c_{p}(x):=\sum_{k=1}^{p} \binom{p}{k} L_{\varepsilon}^{p-k}\hat{s}_{k,x}(0).
\]
The expression in  (\ref{Cnst11}) follows from the Taylor
expansion of $e^{-\lambda_{\varepsilon}t}$, using also that
$\sup_{x}|c_{p}(x)|<\infty$ (a fact which follows from
(\ref{PDv})). Finally, the "constant property" of
(\ref{Cnst11}), for any $0<\varepsilon<\varepsilon_{0}$, follows
from inversion. Indeed, given that a posterior
    \begin{equation}\label{GED2}
 \f_{t}(x)=\sum_{k=1}^{n} \frac{\a_{k}(x)}{k!} t^{k}+ t^{n+1}O_{\eta,\varepsilon}(1),
\end{equation}
holds true for any $t\leq{}1$ and $0<\varepsilon<\varepsilon_{0}$,  $a_{k}(x)$ can be recovered from $f_{t}(x)$ (independently of $\varepsilon$) by the recursive formulas:
\[
    a_{1}(x)=\lim_{t\to{}0}\frac{1}{t} f_{t}(x), \quad a_{k}(x)=\lim_{t\to{}0}\frac{k!}{t^{k}}\left(f_{t}(x)-\sum_{i=1}^{k} \frac{\a_{i}(x)}{i!} t^{i}\right), \quad 2\leq{}k\leq{}n.
\]
\hfill
$\Box$

\bibliographystyle{plain}


\end{document}